\documentclass[12pt]{article}

\usepackage{setspace}
\usepackage{tocloft} 
\usepackage{times}
\usepackage[utf8]{inputenc}
\usepackage{amsmath}
\usepackage{mathtools}
\usepackage{derivative}

\usepackage{mathrsfs}

\usepackage{graphicx}
\graphicspath{ {./one_simulations_result/} }
\usepackage{floatrow}
\usepackage{subcaption}

\usepackage{soul} 
\usepackage{parskip}
\usepackage{tikz}
\usepackage{bm}
\usepackage{multirow}
\usepackage{array}
\usepackage{booktabs}
\usepackage{amsfonts}
\usepackage{amsthm}
\usepackage{tabularx,booktabs}
\usepackage{amssymb}

\usepackage{bbm}

\usepackage{enumitem} 

\usepackage[left=1.in, right=1.in, top=1.in, bottom=1.in, includehead, includefoot]{geometry}

\usepackage{natbib}

\newcolumntype{Y}{>{\centering\arraybackslash}X}
\usepackage[citecolor=blue, colorlinks=true, linkcolor=blue]{hyperref}

\usepackage{changepage} 
\usepackage{setspace} 

\usepackage{pdfpages}

\usepackage{xurl}

\usepackage{authblk}

\theoremstyle{plain}
\newtheorem{definition}{Definition}
\newtheorem{lemma}{Lemma}
\newtheorem{proposition}{Proposition}

\newtheorem{corollary}{Corollary}

\newcommand{\cC}{\mathscr C}
\newcommand{\vV}{\mathcal V}
\newcommand{\EE}{\mathbbm E}
\newcommand{\var}{\mathbbm V}

\title{Trade Networks and the Rise of a Dominant Currency\thanks{We would like to thank our colleagues who gave us useful comments, in particular,	Kozo Ueda, John Stachurski, and Hajime Tomura. This work was supported by JSPS KAKENHI Grant Number JP22H00849. Corresponding author: Lien Pham. Nishi-Waseda Bldg.7F, 1-21-1 Nishi-Waseda, Shinjuku-ku, Tokyo 169-0051 Japan, \texttt{phamlien@ruri.waseda.jp} }}

\author[a]{Tomoo Kikuchi}
\author[b]{Lien  Pham}

\affil[a,b]{Graduate School of Asia-Pacific Studies, Waseda University}

\date{\today}

\begin{document}

\maketitle

\thispagestyle{empty}

\begin{abstract} 
	\noindent 
We develop a model where currency issuers provide liquidity, while users in a trade network choose currency usage for trade settlement. We identify a feedback mechanism where a user's currency preference spillovers to others and increases the issuer's commitment to liquidity provision, which in turn increases the adoption of the currency.  Our findings highlight not only the advantage of the incumbent issuer in maintaining dominance, but also the conditions that lead to the rise and fall of dominant currencies.  Our framework offers testable implications for the share of global settlement currencies, the network structure, and the strategy of issuers.
	
	\vspace{1ex}
	
	\noindent\textbf{Keywords:} global currency; trade network; subgame perfect equilibrium; US-China competition

	\vspace{1ex}
	
	\noindent\textbf{JEL Classification:} C72; D62; D85; E42; F33.

\end{abstract}

\clearpage

\pagenumbering{arabic}

\section{Introduction}

Since the end of World War II, the United States dollar (USD) has been the dominant currency for international payments, invoicing, and foreign reserves.
The USD shapes the cost of trade, the access to international capital, and the vulnerability to US monetary policies, granting the US political power to impose financial sanctions.
Today, the US appears to be in the final stage of a cycle in which global powers rise---first by expanding international trade, then by developing dominant financial centers, and finally by issuing the world’s leading reserve currency \citep{dalio2021changing}. On the other hand, China appears to be in the first phase of this cycle. The rise and fall of dominant currencies
will have far-reaching economic and political consequences. 
 
Since 2013, China has surpassed the US in global trade share, raising questions about whether it could challenge the USD's dominance. To support the international use of the renminbi (RMB), China has offered loans and aids to build infrastructure globally and facilitated trade settlement in RMB through initiatives such as the Belt and Road Initiative (BRI) and the Asian Infrastructure Investment Bank (AIIB), while also signing bilateral currency swap agreements and launching Shanghai crude oil futures. Yet, China
maintains a relatively closed capital account, and many of its initiatives are regional or bilateral---lacking the global scope of institutions built under the Bretton Woods system by the US.
Consequently, despite China’s dominant trade share, the use of the RMB as a global payment currency remains remarkably limited. In April 2025, only 7\% of international trade was settled in RMB, less than one-tenth of 81\% settled in USD.\footnote{Source: SWIFT.} This disparity underscores that a dominant trade volume does not necessarily translate into a strong commitment to global provision of liquidity or into global adoption of the currency.

Under what conditions can we expect the rise and fall of global currencies? 
To answer this question, we integrate the network game introduced in \cite{ballester2006s} into a framework of strategic competition where liquidity provision is costly and currency choice exhibits strategic complementarity---users prefer currencies widely accepted by their trade partners.   
Even if a key trade partner prefers a particular currency, it may not be optimal for a country to adopt it unless it can also be used for settling trade with other partners.
The network game enables us to understand such network-wide incentives rather than bilateral ties alone.
We model the currency choice of users for trade settlement motivated by 
\cite{eichengreen2011renminbi} who observed that currency internationalization has often started by encouraging the use in trade invoicing and settlement. On the supply side, we assume that  
issuers can attract users by increasing liquidity at a cost, and they choose their provisions in sequence, i.e., there is a first mover, a second mover and so on. This setup is natural, for example,  when we would like to analyze how the USD dominance is challenged by other currencies.

Network spillover is at the core of our analysis. 
We show that the more integrated a trade network is, the fewer currencies are used internationally. This is because issuers respond by increasing liquidity, thus deterring the entry of new currencies. 
The central users---those who are well-connected to other users in a trade network---play the key role. 
We show that the more central a country is in a trade network, the more concentrated its currency allocation is. 
It follows that the issuer of a currency preferred by central users increases liquidity and achieves greater adoption. 

This feedback mechanism between the network spillover and the liquidity provision is the key insight we obtain as to why a dominant currency is hard to be replaced, but more importantly, how a challenger will level up the game and cause a sudden shift in dominant currencies when the trade network changes. 
Our main results can be generalized to a market with finitely many currencies.
To this end, we employ a dynamic programming approach, which exploits a recursive structure to obtain a unique solution for the currency choice problem of users. This means that we find a unique optimal allocation in a game that is inherently not recursive by this approach. 
\begin{figure}[ht!]
	\centering
	\begin{subfigure}{0.45\textwidth}
		\includegraphics[width=\textwidth]{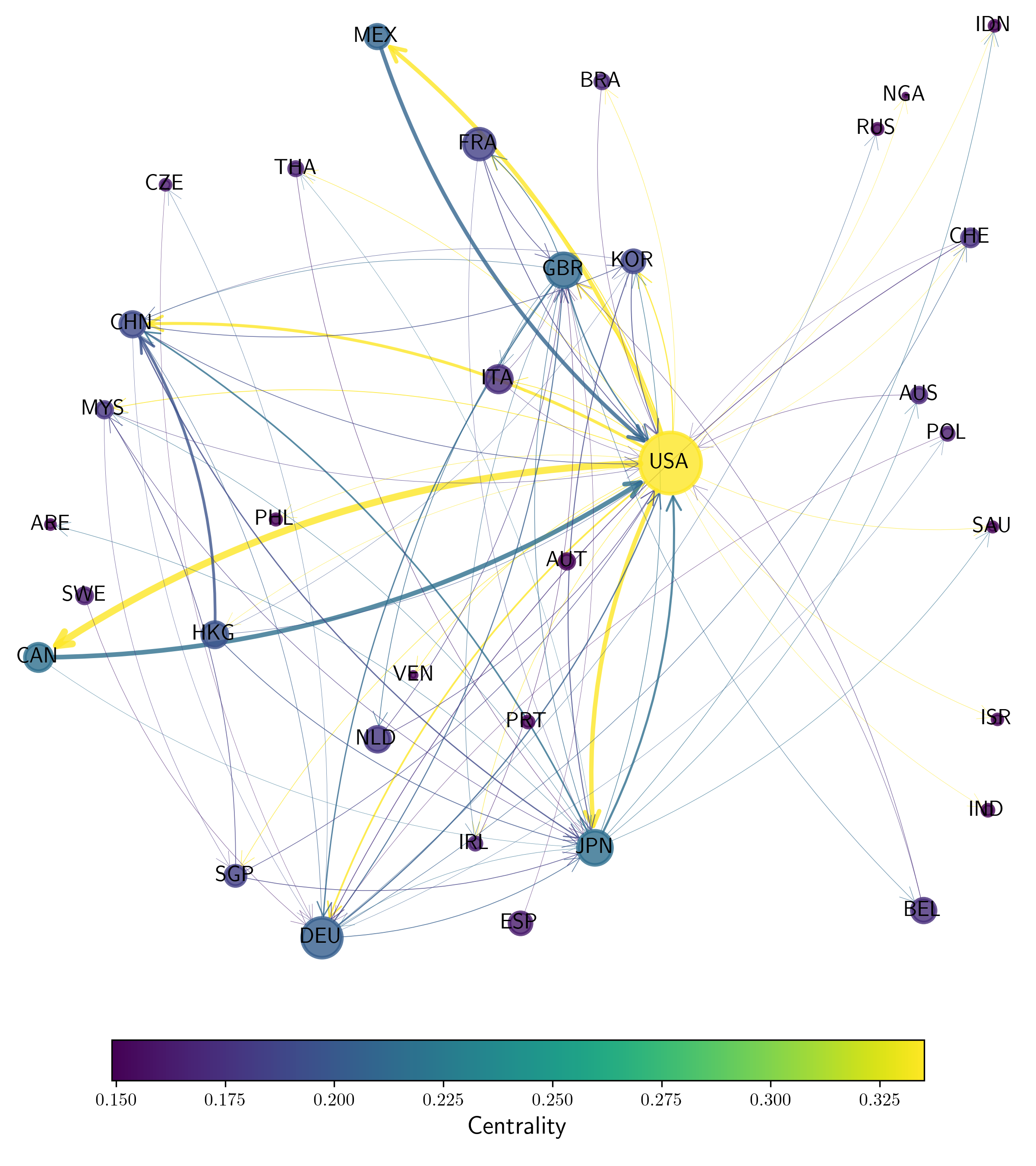}
		\caption{Trade payment network in 2000}
		\label{fig:trade2000}
	\end{subfigure}
	\qquad
	\begin{subfigure}{0.45\textwidth}
		\includegraphics[width=\textwidth]{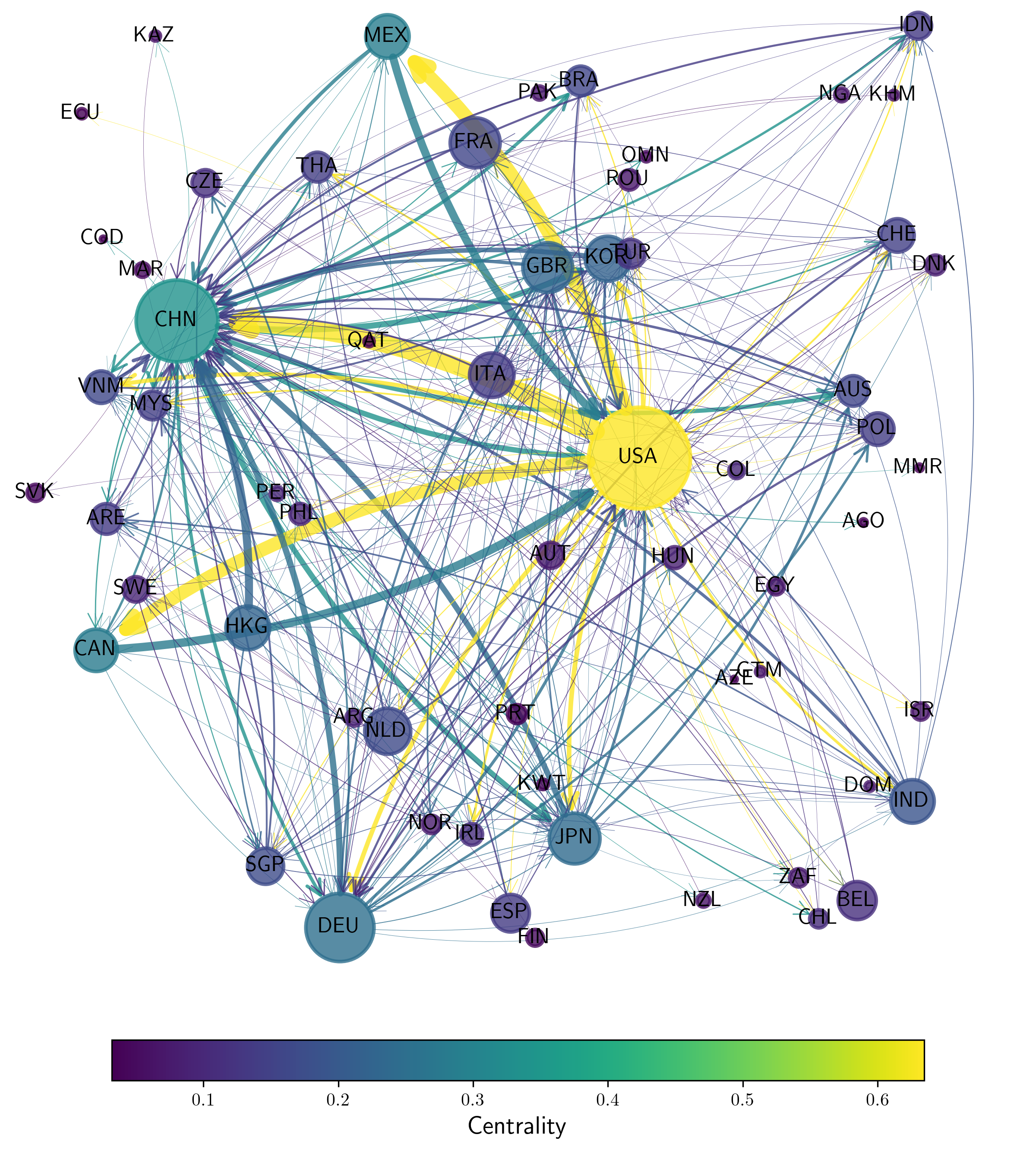}
		\caption{Trade payment network in 2022}
		\label{fig:trade2021}
	\end{subfigure}	
	\caption{Changes in the global trade network: 2000 and 2022}
	\label{fig:trade_nw}
	\parbox{\textwidth} {\footnotesize \textit{Note:} Data source: UN Comtrade. Trade among EU countries is excluded. Only bilateral trade with a value above 10 billion USD is plotted. Each node represents a country, with the node size proportional to total import volume. Directed links indicate payment flows, where a link from country $i$ to country $j$ represents payments made by $i$ (the importer) to $j$ (the exporter) for trade transactions. The thickness of links corresponds to payment volume (i.e., import value). The node colors reflect network centrality, with lighter colors indicating higher centrality. The trade volume is in trillion USD and Katz-Bonacich centrality is computed with $\lambda = 1$ for both years. The formula is introduced in Section \ref{subsec:users}.}
\end{figure}

The feedback mechanism explains why  China does not fully liberalize its capital account and why the RMB's share for trade settlement is so low despite its trade dominance. 
Figure \ref{fig:trade_nw} shows  global trade networks  and the centrality of countries---a measure of how well-connected a country is to other countries---for 2000 and 2021. It
reveals that  
despite its significant growth in bilateral trade, China's centrality increased at a much slower rate than its import volume, while the US retained the highest centrality: The US imports heavily from key importers who also import heavily from other key economies.\footnote{China's top importing partners include Australia, Brazil, Indonesia, Malaysia, Thailand, and Vietnam, whose import volumes are relatively low, and they do not import significantly from other key economies in the network.} 
China's lower centrality, according to our model, decreases its incentive to provide liquidity globally, since users are less responsive to its liquidity provision than to that of the US. In contrast, the high centrality of the US increases its incentive to provide liquidity globally  and amplifies the network spillover reinforcing USD dominance.

Our paper complements recent studies that view international currency competition through the lens of safe asset supply \citep[e.g.,][]{caballero2008equilibrium, maggiori2017financial, gopinath2018trade, gopinath2021banking, clayton2024international}. While these studies emphasize the store-of-value function of currencies and financial risks, our focus is on the medium-of-exchange role and how trade networks shape both currency usage and  internationalization strategies.
Existing theoretical work on settlement currencies often ignores or treats institutional factors such as financial development or capital account liberalization as exogenous \citep[e.g.,][]{rey2001international, chahrour2017international, liu2019trade}, whereas we endogenize  them as strategic commitments to liquidity provision. 
Moreover, network externalities are often captured without  an explicit network structure in the aggregate user-based models \citep[e.g.,][]{vaubel1984government, dowd1993currency} and the search models \citep[e.g.,][]{matsuyama1993toward, lagos2005unified, liu2010currency}. 

More recently, 
\cite{kikuchi2025game} develops a game-theoretic model of two superpowers who provide club goods to compete strategically for their sphere of influence. While the model provides insights into superpower competition and captures externalities among non-superpowers, it is not directly applicable to currency competition, as a currencies--unlike club goods--are generally non-excludable. Moreover, the model does not have an explicit network structure – a crucial factor for determining the currency choice in our model. 

Our model relates to multi-action network games developed by \citet{chen2018multiple} and \citet{chen2022impact}, in which agents engage in multiple activities that may be substitutes or complements.
Agents in their models do not solve a resource allocation problem, which is the key feature of our model. In this respect, the closest to our setup is \citet{belhaj2014competing}, who characterize equilibrium in a network game where each agent allocates fixed resources between two activities. Our setup is more general as we solve an allocation problem for finitely many currencies.

Empirical studies have captured the role of  network effects in determining the choices of currency by applying a logistic transformation to the currency share \citep[e.g.,][]{chinn2008euro, frankel2012internationalization, siranova2020determinants}, using the number of foreign quotations for the currency \citep[e.g.,][]{flandreau2009empirics}, or including lagged currency shares as predictors  \citep[e.g.,][]{bobba2007determinants, chictu2014did, lee2014will}. 
These approaches reflect the idea that countries are more likely to adopt a currency that is already widely used. Building on this, our model provides a testable theoretical prediction of how a country's centrality in the global trade network relates to the choice of settlement currencies. We also show that a country's commitment to currency internationalization---such as financial development and capital account liberalization---is shaped by its network position, alongside political partisanship, trade openness, and legal origin \citep[e.g.,][]{quinn1997origins, porta1998law, beck2003law, rajan2003great}.

The remainder of the paper is organized as follows. Section \ref{sec:model} introduces our model with two currencies. Section \ref{sec:eqm} explains the solution concepts. Section \ref{sec:results} presents the theoretical results. Section \ref{sec:dp} extends the baseline framework to a finite $T$ - currency setting. The proofs are relegated to Appendix.

\section{A two-currency problem}
\label{sec:model}

There are two currency issuers $\{a, b\}$. They issue currency $a$ and currency $b$, respectively.
Each issuer $p \in \{a, b\}$ sequentially decides whether to internationalize its currency and, if so, chooses a \textit{commitment level} $e_p \in [0, 1]$ for liquidity provision. This commitment level can be interpreted as the efforts to manage an open capital account or to maintain institutions, swap lines and infrastructures that expand the market for the currency, thereby enhancing its liquidity.
In this model, $a$ is the first mover and $b$ is the second mover.

There is  a finite set of users $N = \{1,2,3,…,n\}$ $(n \geq 2)$. Each user $i$ must settle $m_i$ units of international transactions with the available currencies. Throughout this paper, we assume that each country has a representative user who decides on the currencies to settle its imports. In our model, it matters little whether the users decide the allocation of currencies to settle their imports or exports (see Section \ref{subsec:users}).

We assume that settling imports with an international currency incurs costs that fluctuate depending on the exchange rates. These costs are mitigated by two factors: the liquidity of the currency and  the network externalities of using the same currency as trade partners. We assume that the liquidity of a currency is determined by the commitment level of the issuer. 
The objective of users is to minimize the cost of import settlement with the two international currencies.

We assume that trade relationships are given exogenously. Countries (importers) are naturally more influenced by the currency choices made by their dominant trade partners (exporters). To capture this, we represent the network as a weighted directed graph, where links represent the flows of payments, and therefore originate from importers to exporters.
Each link $ij$ from $i$ to $j$ with $(i, j) \in N^2$ has a weight $w_{ij} \geq 0$, indicating how much $i$ values $j$'s decision, or the dependency of $i$ on $j$. When $w_{ij} > 0$, $j$'s usage and $i$'s usage of the same currency are strategic complements from $i$'s perspective. For simplicity, we assume that $w_{ii}=0$ for all $i \in N$. The adjacency matrix for the network is $\bm{w} = [w_{ij}]$, which is a zero-diagonal non-negative square matrix.

The game unfolds in three stages:

\begin{enumerate}
	\item[(i)] The incumbent $a$ chooses whether to internationalize its currency and, if so, its commitment level $e_a \in [0, 1]$ for liquidity provision. 
	
	\item[(ii)] Given $a$'s choice, the challenger $b$ chooses whether to internationalize its currency and, if so, its commitment level $e_b \in [0, 1]$ for liquidity provision.  
	
	\item[(iii)] Given the choices of the two issuers, each user $i$ decides the usage of each currency, $x_{ia}$ and $x_{ib}$, to settle its imports such that $x_{ia}+x_{ib} = m_i$ for all $i \in N$.
\end{enumerate} 

The cost for user $i$ to settle $x_{ip}$ of its transactions with currency $p$ is equal to the import volume in currency $p$ multiplied by the unit cost of settling imports in currency $p$ 
given by  
\[
C_{ip} (\bm{x_p}, e_p) 
=  x_{ip} \bigg[
	\underbrace{ \vphantom{\sum_{j \neq i }}
	\epsilon}_{\text{exchange rate risk}}
	 - \underbrace{ \sum_{j \neq i } w_{ij} x_{jp}}_{\text{network externality}}
	 - \underbrace{ \vphantom{\sum_{j \neq i }} 
	 	f_{ip}(e_p)}_{\text{liquidity}}
 	\bigg],
\] 
where $\bm{x_p}$ denotes the vector of usage of  currency $p$ by all users and  the unit cost term, $[\epsilon_p - \sum_{j \neq i } w_{ij} x_{jp} - f_{ip}(e_p)]$, captures risks and fees arising from using currency $p$. The term $\epsilon$ is a random variable with variance $\sigma^2$, which represents the costs of acquiring and exchanging the currency that fluctuates with exchange rates of the currency. For simplicity, we assume that $\epsilon$ has a mean of zero and a common variance for the two currencies. The network externality term, $\sum_{j \neq i } w_{ij} x_{jp}$, reduces the unit cost when user $i$'s trade partners also use $p$. Finally, $f_{ip}(e_p)$ captures the liquidity of currency $p$, which reduces the unit cost as the currency may be used for other purposes such as asset holdings for store of value. We assume that $f_{ip}(e_p): [0,1] \rightarrow [0, \infty)$ is a differentiable, concave, and strictly increasing function of the commitment level $e_p$. 

We assume that users have {\em mean - variance preferences}, that is, they dislike the risk measured by the variance of the cost. Hence, the user $i$'s perceived cost of settling $x_{ip}$ transactions with currency $p$ becomes
\begin{align*}
	\cC_{ip} (\bm{x_p}, e_p) \coloneqq \frac{\eta}{2} \var[C_{ip}] + \EE[C_{ip}] 
	= \frac{\beta}{2} x_{ip}^2
	- \sum_{j \neq i } w_{ij} x_{ip} x_{jp}
	- f_{ip}(e_p) x_{ip},
\end{align*}
where $\eta >0$ represents the risk aversion, which is assumed to be the same for all users, and $\beta \coloneqq \eta \sigma^2$ captures the exchange rate risk perceived by users. As will be shown in Section \ref{subsec:users}, $\beta$ is assumed to be sufficiently large to ensure a unique and non-negative solution for the currency choice problem. 
The convex cost associated with the exchange rate risk gives users an incentive to diversify their currency use.\footnote{Although this functional form is restrictive, such quadratic functions can be seen as the second-order Taylor approximations of general convex functions that represent risk aversion, providing a balance between theoretical clarity and analytical tractability. This approach aligns with foundational work in financial economics \citep[e.g.,][]{markowitz1952portfolio} and systemic risks \citep[e.g.,][]{acemoglu2015networks}, while enabling closed-form solutions that connect equilibrium outcomes to network structure.}

Each user $i$ decides its usage of the two international currencies issued by $a$ and $b$ by solving 
\begin{equation}
	\label{value_both}
	\min_{0 \leq x_{ia}, x_{ib} \leq m_i} v_i (\bm{x_a}, \bm{x_b},  e_a,  e_b ) = \cC_{ia} +  \cC_{ib}  
	\quad \text{s.t. } x_{ia} + x_{ib} = m_i.
\end{equation}
Each issuer $p\in \{a,b\}$  solves the utility maximization problem given by
\begin{equation} \label{max_u}
	\max_{e_p} u_p(e_p, e_{-p}) = 
	\begin{cases}
		X_p(e_p, e_{-p}) - k c(e_p)   & \mbox{if  $0 \leq e_p \leq 1$} \\
		0 & \mbox{if $ e_p = \emptyset$,}
	\end{cases}
\end{equation}
where $X_p \coloneqq \sum_{i \in N} x_{ip}$ is the total usage of currency $p$, the cost for liquidity provision $c(e_p): [0,1] \rightarrow (0, \infty)$ is differentiable, convex, and strictly increasing in $e_p$, and $k\geq0$ captures the weight of the commitment cost. For simplicity, we assume that two issuers face the same cost function. Note that $c(0)>0$ because even with zero commitment, liquidity provision  requires the establishment and maintenance of  financial institutions and legal frameworks. Hence, we assume that it is costlier to internationalize a currency with $e_p = 0$ than not to internationalize it at all, i.e.,  $e_p = \emptyset$. 

\section{The equilibrium}
\label{sec:eqm}
We find the subgame perfect Nash equilibrium (SPE) by solving the problem backward. We begin this section with the currency choice problem of users given the commitment to liquidity provision by issuers.

\subsection{The users}
\label{subsec:users}

Each user minimizes its cost given the commitment to liquidity provision  by issuers and the currency choices made by  other users.
We assume that when there is only one issuer internationalizing its currency, users settle all their transactions in the only available international currency; and if both issuers give up on internationalizing their currencies, users cannot settle their imports using any international currency.\footnote{If neither issuer internationalizes their currency, trade must be settled in local currencies, a scenario not modeled in this paper as we focus on cases where at least one international currency is available.} We first derive the currency choices by users given $e_a$ and $e_b$ in the range [0,1], that is, when both issuers internationalize their currencies.

Since $x_{ib} = m_i - x_{ia}$ for all $i \in N$, the FOC  for (\ref{value_both}) with respect to $x_{ia}$ yields
%
\begin{align} \label{foc:both_a}
	x_{ia} ({\bm{x_a}, \bm{x_b}, e_a, e_b }) =  \frac{1}{2 \beta } \left[ \sum_{j \neq i } w_{ij} (x_{ja} -  x_{jb})+f_{ia}(e_a) -  f_{ib}(e_b) + \beta m_i
	\right].
\end{align}
Equation (\ref{foc:both_a}) shows that 1) given any pair $(e_a, e_b)$, the currency choice by user $i$ depends on that of other users, and 2) the usage of currency $a$ by each user is increasing in the liquidity of currency $a$, $f_{ia}(e_a)$, and decreasing in that of currency $b$, $f_{ib}(e_b)$.
Similarly, the FOC with respect to $x_{ib}$ yields
\begin{align} \label{foc:both_b}
	x_{ib} ({\bm{x_a}, \bm{x_b}, e_a, e_b }) =  \frac{1}{2 \beta } \left[ \sum_{j \neq i } w_{ij} ( x_{jb}-x_{ja}) + f_{ib}(e_b)- f_{ia}(e_a) + \beta m_i
	\right].
\end{align}
Let $d_i \coloneqq x_{ia} - x_{ib}$ be the \textit{usage difference} chosen by user $i$ in the equilibrium given $a$ and $b$'s commitment levels.
From (\ref{foc:both_a}) and  (\ref{foc:both_b}), for all $i \in N$
\begin{eqnarray}  
d_i&=& \frac{1}{\beta} \left[\sum_{j \neq i } w_{ij}(x_{ja} - x_{jb}) + f_{ia}(e_a) - f_{ib}(e_b)  \right]  \label{subs}\\
&=&\frac{1}{\beta} \sum_{j \neq i } w_{ij}d_j+\gamma_i,  \label{subs_short}
\end{eqnarray}
where $\gamma_{i}(e_a, e_b) \coloneqq \frac{ f_{ia}(e_a) - f_{ib}(e_b)}{\beta}$.
We refer to $\gamma_i(e_a, e_b)$ as the \textit{difference in marginal costs} that user $i$ faces when choosing between the two currencies. Note that the numerator of this term is equal to the difference in marginal costs between using currency $b$ versus using currency $a$, excluding the network effects. The denominator normalizes this difference by the user's sensitivity to the exchange rate risks. Thus, $\gamma_i$ quantifies the non-network-driven marginal preference for currency $a$ over currency $b$.
Given the vector $\bm{\gamma}$ with entries $\gamma_i$, the vector $\bm{d}$ with $d_i = x_{ia} - x_{ib}$ that satisfies (\ref{subs_short}) can be written as $\bm{d} =  \frac{1}{\beta} \bm{w} \bm{d} + \bm{\gamma}$. Solving it for $\bm{d}$ yields
\begin{align} \label{foc:d}
	\bm{d} & = (\mathbb{I} - \frac{1}{\beta} \bm{w})^{-1} \bm{\gamma}	
	 = \sum_{\ell=0}^\infty \frac{1}{\beta^\ell} \bm{w}^\ell \bm{\gamma},
\end{align}
where $\mathbb{I}$ is the identity matrix.
Let $r(\bm{w})$ be the spectral radius of $\bm{w}$, i.e., the absolute value of the largest eigenvalue of $\bm{w}$. In order for vector $\bm{d}$  to be unique and finite, we need $\beta > r(\bm{w})$. 
For any $i \in N$, we have $x_{ia} + x_{ib} = m_i$ and $x_{ia} - x_{ib} = d_i$. Therefore, the optimal usage level of each currency by user $i$ is
\begin{equation} \label{foc:usage}
	x_{ia} = \frac{m_i + d_i}{2} \quad\text{and}\quad x_{ib} = \frac{m_i-d_i}{2}.
\end{equation}
Define $-p \coloneqq b$ if $p=a$, and $-p \coloneqq a$ if $p=b$. From (\ref{subs}), $- m_i \leq d_i \leq m_i $ for all $i$ and for any pair $(e_a, e_b)$ when $\beta \geq \max_{i \in N, p \in \{a,b\}} \frac{\sum_{j \neq i} w_{ij} m_j + f_{ip}(1) - f_{i,-p}(0)}{m_i}$. Hence, when $\beta$ is sufficiently large, the resulting usage levels of the two currencies are non-negative.
Let $\underline{\beta}$ be the minimum values of $\beta$ that ensure a unique and nonnegative solution for the  currency choice problem. Throughout this paper, we assume that $\beta\geq \underline{\beta}$.
So far, given any pair $(e_a, e_b)$, we have derived the unique solution for the currency choice problem.

It is important to note that (\ref{foc:d}) implies that a user's usage level of each currency is closely related to its Katz-Bonacich centrality in the network \citep{katz1953new, bonacich1987power}. 
Given a parameter $\lambda$ that satisfies $0<\lambda<\frac{1}{r(\bm{w})}$, the vector of Katz-Bonacich centralities of the nodes in network $\bm{w}$ is finite and uniquely defined by
\begin{align} 
	 \bm{\kappa} & \coloneqq (\mathbb{I} - \lambda \bm{w})^{-1} \bm{1}	 = \sum_{\ell=0}^\infty \lambda^\ell \bm{w}^\ell \bm{1}			 \label{katz},
\end{align}
where $\bm{1}$ is an $n \times 1$ vector of ones. Here, each entry $A_{ij}$ of the square matrix $\bm{A} \coloneqq \sum_{\ell=0}^{\infty} \lambda^{\ell} \bm{w}^{\ell}$ counts all the paths in $\bm{w}$ that start at $i$ and end at $j$, with paths of length $\ell$ being weighted by $\lambda^\ell$. Therefore, the Katz-Bonacich centrality of user $i$, $\kappa_i$, measures the total number of paths originating from $i$, with each path weighted by a decay factor that diminishes as the path length increases.
A user has a high Katz-Bonacich centrality if 1) many links leave the user, 2) these links have large weights, and 3) these links point to users with high centralities.
From (\ref{foc:d}) and (\ref{katz}), we have the following result.
\begin{lemma}
	\label{lem:d_kappa}
	Define the Kat-Bonacich centrality $ \bm{\kappa}  \coloneqq (\mathbb{I} - \frac{1}{\beta} \bm{w})^{-1} \bm{1}	$. When $\gamma_i(e_a, e_b) = \gamma$ for all $i \in N$, we have $\bm{d}(e_a, e_b) = \gamma \bm{\kappa} $. When $\gamma_i(e_a, e_b)$ differs across users, by defining an adjusted Katz-Bonacich centrality $\bm{\tilde{\kappa}} \coloneqq (\mathbb{I} - \frac{1}{\beta} \bm{w})^{-1} \bm{\gamma}$, we have $\bm{d}(e_a, e_b) = \bm{\tilde{\kappa}}$.
\end{lemma}

Lemma \ref{lem:d_kappa} states that the difference in the usage of the two currencies by each user is proportional to its (adjusted) Katz-Bonacich centrality. When $\gamma_i(e_a, e_b) > 0 ~\forall i$, a user with a higher centrality uses more of currency $a$ and less of currency $b$. In contrast, when $\gamma_i(e_a, e_b) < 0 ~\forall i$, vise versa. Intuitively, when a user becomes more central in the network, its currency use becomes more uneven. 
This is because a more central user benefits more from matching the currency with its trade partners. Moreover, a central user's currency choice influences others (directly and indirectly), which in turn raises the central user’s benefit of using the currency. Thus, a user with a higher centrality has the incentive to use the currency that it prefers more. 

Since the Katz-Bonacich centrality is calculated based on paths originating from each node, the direction of links, that is, how users take others' behavior into account, is important in determining the currency choice. We believe that in a currency choice problem, the directed links should represent the flows of payments and therefore, originate from the importers to the exporters.
This is because an importer prefers to pay the exporter in a currency commonly used by both of them to minimize exchange costs. Similarly, an exporter prefers to receive payments in a currency it can use to pay its own exporters, reducing future conversion needs. In both cases, users make currency decisions based on their payment destinations. Thus, the direction of the links reflects how much weight users place on each other’s currency choices, with the flows of payments naturally representing these interdependent decisions.
Therefore, no matter whether we let the importers or the exporters decide the transaction settlement, their decisions are still proportional to the (adjusted) Katz-Bonacich centralities in terms of the payment flows. Hence, it matters little whether the users decide the currencies to settle their imports or exports in our model. For simplicity, we assume that the users decide the currencies to settle their imports and thus $m_i$ is user $i$'s total import.\footnote{If we assume that the users decide the currencies to settle their exports, the only change in the solution for the currency choice problem is that $m_i$ becomes user $i$'s total export.}

\subsection{The issuers}

Given the responses of users, the payoffs to the two issuers for their commitment levels $(e_a, e_b)$ are well-defined. We make the following two assumptions on the choices of issuers.
\begin{enumerate}[label=(\roman*)]
	\item \label{as:zero} If $u_p \leq 0$, then $p$ does not internationalize its currency. If $u_p > 0$, then $p$ internationalizes its currency with a commitment  level $e_p \in [0,1]$ for liquidity provision.
	\item \label{as:open} Consider two distinct choices $e_p$ and $e'_p$. If $u_p (e_p) = u_p(e'_p)$, then $p$ chooses $\max \{e_p, e'_p \}$.
\end{enumerate}	
Assumption \ref{as:zero} states that an issuer does not internationalize its currency unless it is beneficial for itself. This assumption is essential for the existence of the Nash equilibria.
Assumption \ref{as:open} states that the issuers choose a higher commitment level for liquidity provision  if multiple choices yield the same payoff. This assumption allows us to investigate the best that issuers would do to internationalize their currencies.
Although the players have infinitely many actions to consider, we can show that the best responses exist at every subgame, and thus the SPE exists.
\begin{lemma} \label{lemma:SPE}
	Suppose that assumption \ref{as:zero} holds. Then, the game has at least one subgame-perfect Nash equilibrium.
\end{lemma}
Therefore, by backward induction, we obtain the SPE. By assumption \ref{as:open}, ties are broken and we can choose a unique SPE where issuers do their best to internationalize their currencies.

Let $M \coloneqq \sum_{i \in N} m_i$ be the total demand for international currencies.
It is worth noting that when $k > \frac{M}{c(0)}$, the cost of the lowest commitment level is higher than the benefit from obtaining the whole market, thus no issuer has an incentive to internationalize its currency. Therefore, we focus on the range of $k$ such that $0 \leq k \leq \frac{M}{c(0)}$ to analyze the SPE. In general, under any game setting, there is a threshold $\underline{k}$ that is crucial for determining the distribution of currency usage in the equilibrium outcome. Denote $e_p^*$ as the choice of issuer $p$ in the SPE outcome. We have the following result. 

\begin{lemma} \label{lemma:k}
	Given any $\bm{w}$, $\bm{m}$ and $\beta$, there exists  $\underline{k} \geq 0$ such that
	\begin{itemize}
		\item[(i)] If $0\leq k < \underline{k}$, then both issuers internationalize their currencies, i.e., $e_a^* \geq 0$ and $e_b^* \geq 0$.
		\item[(ii)] If $\underline{k} \leq k \leq \frac{M}{c(0)}$ , then one issuer attains its monopoly, i.e., $\exists p \in \{a,b\}$ such that $e_p^* \geq 0$ and $e_{-p}^* = \emptyset$.
	\end{itemize}
\end{lemma}
Lemma \ref{lemma:k} states that given any game setting, when $k$ is sufficiently low, both issuers choose to internationalize their currencies; while when $k$ is sufficiently high, there is only one international currency issuer in the market.

\section{Theoretical results}
\label{sec:results}
\subsection{The incumbent advantage}
\label{symmetric_case}

To highlight the incumbent advantage, we assume that the players are ``symmetric" in this section and study the equilibrium outcome.

\begin{definition}
The users are symmetric if they benefit equally from liquidity of two currencies, i.e., $f_{ia}(e) = f_{ib}(e) = f(e) ~\forall i \in N$.
\end{definition}

Let  $e_p^*$ be the choice of issuer $p$ and $X_p^*$ be the total usage of currency $p$ in the SPE outcome. 
We obtain the following result.
\begin{proposition} \label{prop:1stmover}
 Suppose that the users are symmetric. Given any $\bm{w}$ and $\bm{m}$, there exist $\underline{k} $ and $\bar{k}$ such that $0<\underline{k} < \bar{k} < \frac{M}{c(0)}$ and
	\begin{itemize}
		\item[(i)] If $\bar{k} \leq k \leq \frac{M}{c(0)}$, then  $a$ attains its monopoly with zero commitment level, i.e., $ e_a^* = 0 , e_b^* = \emptyset, X_a^* = M$ and $X_b^* = 0$.
		\item[(ii)] If $\underline{k} \leq k < \bar{k} $,  then $a$ attains its monopoly with a positive commitment level, i.e., $0 < e_a^* \leq 1 , e_b^* = \emptyset, X_a^* = M$ and $X_b^* = 0$.
		\item[(iii)]If $0 \leq k < \underline{k}$, then $a$ and $b$ choose the same commitment level and share the market equally, i.e., $0 \leq e_a^* = e_b^* \leq 1$ and $X_a^* = X_b^* = \frac{M}{2}$.
	\end{itemize} 
Moreover, $\underline{k} $ is increasing in $\beta$ and $M$.
\end{proposition}


Given that the users are symmetric, when issuers choose the same commitment level, each user's position in the network does not make them favor any currency. Since the two currencies are identical to all users, no matter how important a user is in the network, it chooses to use both currencies equally. That is, $d_i = \gamma_i \kappa_i = 0$ for any user $i$ in this case (see Lemma \ref{lem:d_kappa}). Therefore, users' centralities do not affect their currency choice as long as the issuers choose the same commitment level. In other words, the network structure is neutral to the results under this setting.

\begin{figure}[ht!]
	\centering
	\includegraphics[width=0.5\textwidth]{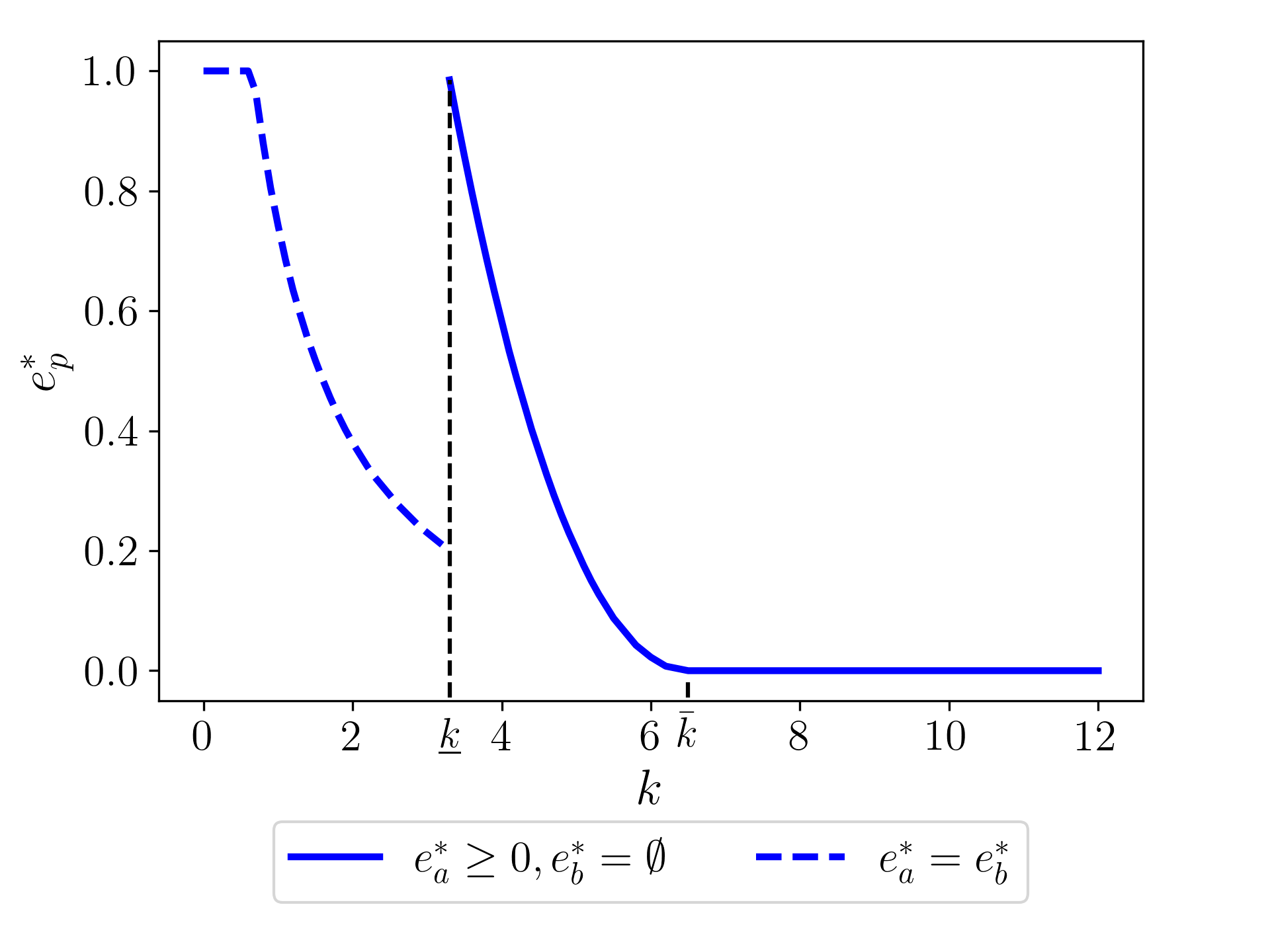}
	\caption{Commitment levels of issuers}
	\label{fig:sym}
\end{figure}

Figure \ref{fig:sym} depicts an example of the equilibrium commitment levels as a function of $k$, where the user network is generated randomly.\footnote{$n=12, m_i = 1 ~\forall i \in N, \beta=2.2, f(e) = e^{0.5}$, $k = 1$, and $c(e) = \exp{(e)}$.}
Intuitively, when the cost of liquidity provision is very high relative to the benefit from global usage, the incumbent can deter the challenger with a low commitment level.  This is because the incumbent issuer anticipates that even without a strong commitment, the challenger would incur a loss if it decided to enter the market due to the high costs involved. 
As the cost of internationalizing a currency decreases, the incumbent must raise its commitment level to deter the challenger. The reduced cost lowers the barrier for the challenger to enter the market, thereby increasing its incentive to internationalize its currency. Consequently, the incumbent needs to commit more substantially to maintain its dominant position.
When the cost is sufficiently low, both issuers have the capacity to internationalize their currencies with high commitment levels. In this scenario, the benefit of global currency usage outweighs the cost for both issuers, leading to a more competitive environment where multiple currencies can coexist internationally.

This proposition suggests that a multipolar payment system may emerge if the cost of liquidity provision becomes sufficiently low relative to the potential benefits. Historically, internationalizing a currency has been costly---requiring opening capital accounts, establishing an international payment system, and building supporting institutions. The US made great efforts to promote the use of its currency worldwide as a first mover, including the establishment of the Bretton Woods Institutions, the World Trade Organization, the petrodollar system, etc.  However, the cost of internationalizing is arguably decreasing as  advancements in payment technologies can reduce monitoring costs for issuers while enhancing security and transparency. As access to these technologies improves, the costs of liquidity provision decrease, making the benefits of enhanced international status more pronounced for issuing countries.

Proposition  \ref{prop:1stmover} also states that the threshold $\underline{k} $ that determines the adoption of the two currencies is increasing in $\beta$. This is because when both currencies are risky, it is harder for the incumbent to maintain its monopoly as users prefer to diversify.
Moreover, when the demand for international currencies $M$ grows, the commitment cost becomes lower than the benefit that the issuers can gain from liquidity provision. Therefore, it becomes harder for the incumbent to deter the challenger and thus the threshold $\underline{k}$ increases.
The change in the market share of the two issuers as $\beta$ and $M$ change is depicted in Figure \ref{fig:box_beta} and Figure \ref{fig:box_m}, respectively.

\begin{figure}[ht!]
	\centering
	\begin{subfigure}{0.49\textwidth}
		\includegraphics[width=\textwidth]{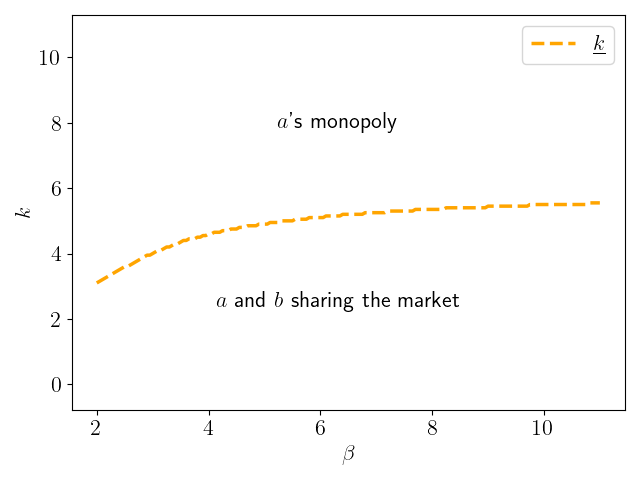}
		\caption{Market share as $\beta$ changes}
		\label{fig:box_beta}
	\end{subfigure}
	\hfill
	\begin{subfigure}{0.49\textwidth}
		\includegraphics[width=\textwidth]{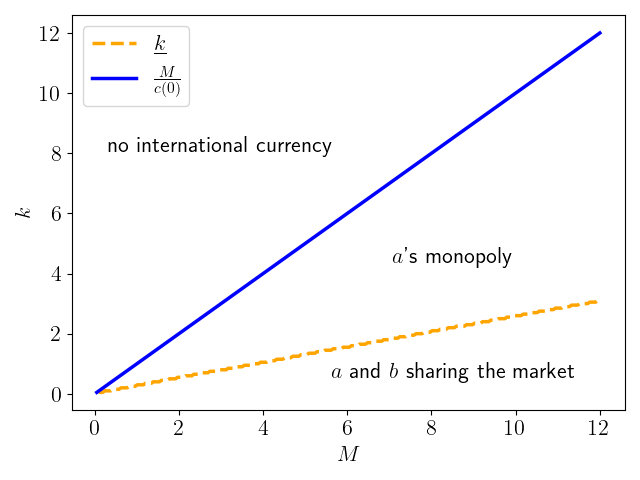}
		\caption{Market share as $M$ changes}
		\label{fig:box_m}
	\end{subfigure}	
	\caption{Distribution of the market share}
	\label{fig:box}
\end{figure}

It is worth noting that if $a$ and $b$ choose their commitment levels simultaneously, then when $0 \leq k < \underline{k}$, the two issuers choose the same commitment level and share the market equally in
 the unique equilibrium. However, when $\underline{k} \leq k \leq \frac{M}{c(0)}$, there are two pure-strategy equilibria, in which either issuer internationalizes its currency. The sequential game enables us to refine the equilibria and select a unique outcome where the incumbent  deters the emergence of a new global currency. This refinement is reasonable in the context of currency competition, as the incumbent enjoys a natural advantage: it benefits from established trust, built financial infrastructures, deep and liquid asset markets, and entrenched network effects that make it costly for users to switch. These advantages allow the incumbent issuer to preempt potential challengers by committing just enough to maintain dominance without inviting competition.
 
\subsection{Home bias and centrality}
 
It is reasonable to assume that some users are not neutral regarding the choice of international currencies for trade settlement. For example, the US users would prefer the USD to any other currencies since they can settle both domestic and international transactions in the USD, thereby reducing the cost associated with the exchange rate risks. We capture these asymmetries with the following definition.
\begin{definition}
 	User $i$ prefers currency $p$ over $-p$ if $f_{ip}(e) = \mu_i f_{i,-p}(e)$ with  $\mu_i > 1$.
\end{definition}

While the centrality considering paths originating from each user plays a crucial role in currency choice (see Lemma \ref{lem:d_kappa}), the cost function suggests that a user with many paths pointing toward itself has greater influence over others' currency choices. To simplify the analysis, we assume that network links are two-way and symmetric, meaning that  $w_{ij} = w_{ji}$ for any pair $(i, j)$. Under this assumption, the network $\bm{w}$ is undirected, ensuring that the effects of outgoing and incoming connections on  currency choices are balanced. To show how issuers compete depending on network positions of home users, we focus on the range of $k$ where the incumbent does not have an advantage (that is, $k<\underline{k}$). Let $\bm{\kappa}  \coloneqq (\mathbb{I} - \frac{1}{\beta} \bm{w})^{-1} \bm{1}$ be users' Katz-Bonacich centrality vector. We have the following result.
 
\begin{proposition} \label{prop:katz}
Suppose that there are two distinct users $i$ and $j$ with opposite preferences, i.e.,  $f_{ip}(e) = \mu f_{i,-p}(e)  = f_{j,-p}(e) = \mu f_{jp}(e) $  where $\mu > 1$. 
Let all other factors be symmetric: $\bm{w}$ is undirected, $m_h = m ~\forall h \in N$, and $f_{ha}(e) = f_{hb}(e) = f(e)$ for all  users $h \in N \backslash \{i, j\}$.
If $\kappa_i > \kappa_j$ and $k < \underline{k}$, then $e^*_p \geq e^*_{-p}$ and $X^*_p > X^*_{-p}$.
\end{proposition}
 
Under the conditions that a) interdependencies between users are two-way and symmetric, b) there are two users with opposite preferences, and c) cost is sufficiently low so that both issuers decide to internationalize their currencies, Proposition \ref{prop:katz} states that the issuer of the currency preferred by the user with a higher centrality commits more strongly and captures a larger market share. The intuition is as follows. First, since $\kappa_i > \kappa_j$, we have $x_{ip} > x_{j,-p}$ given the same commitment by the two issuers. Since the bilateral weight is symmetric, the total cross-effect that user $i$'s decision has on other users must be larger than that of $j$, giving the issuer of the currency favored by $i$ (i.e., issuer $p$) a higher incentive to raise its commitment. Due to this feedback mechanism, in the SPE outcome, $p$ chooses a higher commitment than $-p$ and obtains a larger market share. Therefore, the network centrality of the users with a stronger preference for a particular currency plays an important role in shaping the currency competition.
 
Since the home currency is the primary medium of exchange for domestic transactions,  a larger domestic economy implies higher liquidity of the home currency for the home-country user. Thus, a larger economic size relative to trade incentivizes users to adopt their home currency for both domestic and cross-border transactions, reducing exposure to exchange rate risk and hedging costs. The US has a significantly higher GDP-to-trade ratio (3.84 in 2022) than China (2.7 in 2022), implying a stronger home bias among US users. Given that both home bias and centrality are higher for the US, our model suggests that China has weaker incentives to internationalize its currency, leading to its limited adoption.

\subsection{Network Integration}

In this section, we investigate how network integration influences the outcomes of the game. For any two networks $\bm{w}$ and $\bm{w'}$, we say $\bm{w'}$ is strictly more integrated than $\bm{w}$ and write $\bm{w'}  \gneq \bm{w}$ if $w'_{ij} \geq w_{ij}$ for all $(i,j )\in N^2$ with at least one strict inequality. 
Let $\underline{\beta}(\bm{w})$ denote the minimum values of $\beta$ that ensure a unique, finite, and non-negative usage allocation by the users given network $\bm{w}$ and any pair $(e_a, e_b)$. Throughout this section, we assume that $\beta \geq \max \{\underline{\beta}(\bm{w}), \underline{\beta}(\bm{w'})\}$. 
We have the following result.

\begin{lemma} \label{lemma:dX}
	Suppose that $\bm{w'}  \gneq \bm{w}$. Given any $\bm{m}$ and $\beta$,   
	\begin{itemize}
		\item[(i)] 
		Commitments are more effective under $\bm{w'}$, i.e.,  $\frac{\partial X_p (\bm{w'})}{\partial e_p}  >  \frac{\partial X_p (\bm{w})}{\partial e_p}$ for any $p \in \{a,b\}$.
		\item[(ii)] Suppose that all users prefer currency $p$. Then, the increase in commitment effectiveness due to network integration is larger for $p$, i.e., $\frac{\partial X_p (\bm{w'})}{\partial e_p}  - \frac{\partial X_p (\bm{w})}{\partial e_p}  >  \frac{\partial X_{-p} (\bm{w'})}{\partial e_{-p}} -  \frac{\partial X_{-p} (\bm{w})}{\partial e_{-p}} $. 
	\end{itemize}
\end{lemma}

Lemma \ref{lemma:dX} states that network integration increases users' response to issuers' commitment. Specifically, an increase in an issuer's commitment results in a larger increase in the usage of its currency when the network is more integrated. This network integration effect is amplified for a currency that is preferred by users.

In the following arguments, we study the influence of network integration when players are symmetric  as  in Section \ref{symmetric_case}. We show that when the network is more integrated, $a$  can become more committed, deter $b$, and be the sole global currency issuer.
\begin{proposition} \label{prop:w} Suppose that the users are symmetric. Given any $\bm{m}$, if $\bm{w'}  \gneq \bm{w}$, then $\underline{k}(\bm{w'}) < \underline{k}(\bm{w})$.
\end{proposition}
Proposition \ref{prop:w} implies that network integration amplifies the incumbent advantage. 
Suppose $k$ is sufficiently low such that two issuers internationalize their currencies initially. As users become more interconnected, the incumbent has the incentive to be more committed  and as a result, it can deter the challenger and attain its monopoly. The effect of network integration when all 
users are symmetric
is illustrated in Figure \ref{fig:w}.\footnote{$n=12, m_i = 1 ~\forall i \in N, \beta=2.2, f(e) = e^{0.5}$, $k = 1$, and $c(e) = \exp{(e)}$.}
\begin{figure}[ht!]
	\centering
	\includegraphics[width=0.5\textwidth]{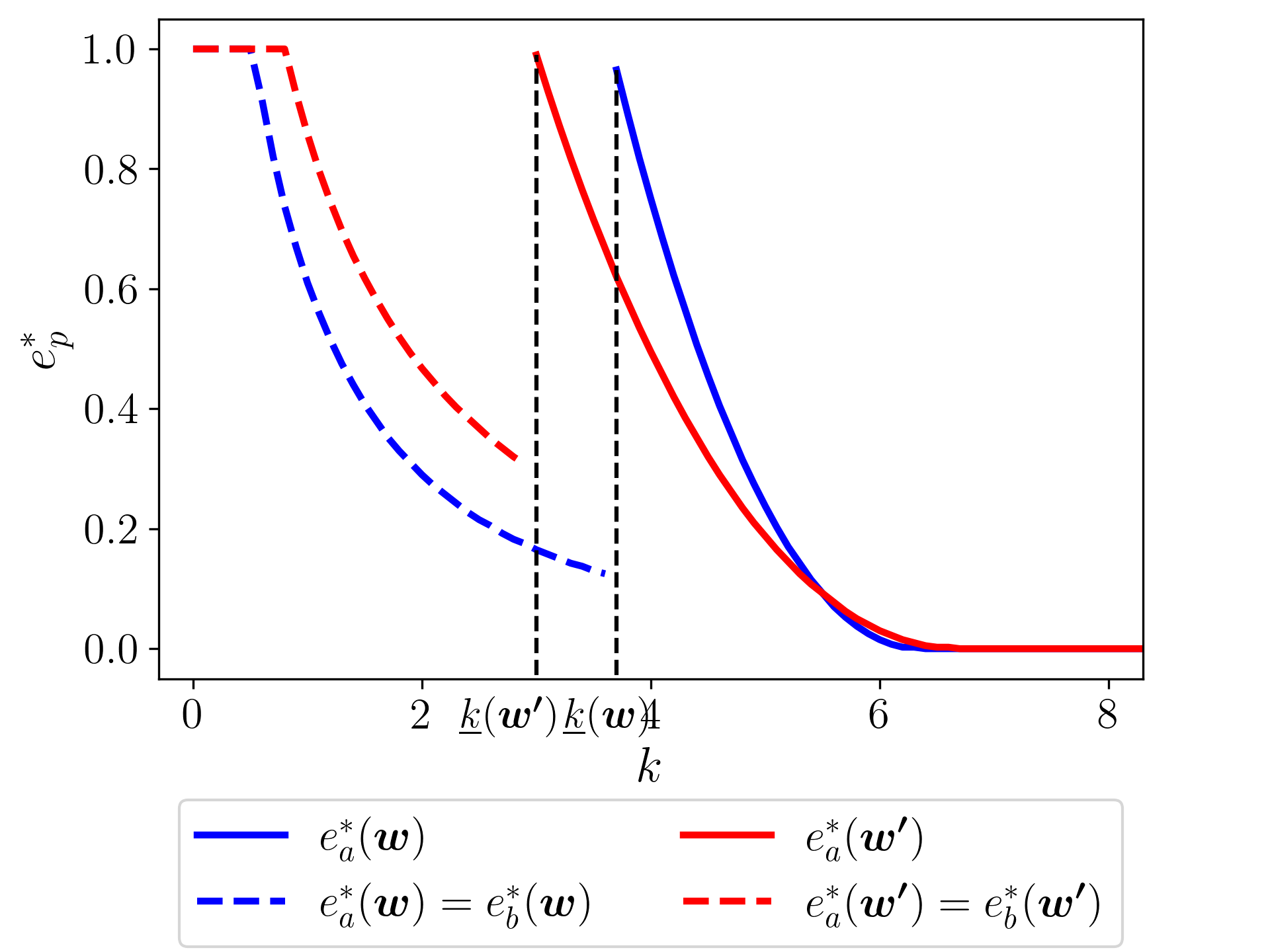}
	\caption{The effect of network integration on commitment levels: $\bm{w'} \gneq \bm{w}$}
	\label{fig:w}
\end{figure}

Lemma \ref{lemma:dX} and Proposition \ref{prop:w} imply that there are two effects that network integration has on issuers' behavior and the currency adoption in the equilibrium. First, since users respond more to issuers' commitment, network integration gives issuers, especially the issuer of the currency preferred by users, more incentive to provide liquidity. Second, it amplifies the advantage held by the incumbent.
When users have a preference for the emerging currency over the incumbent, which of these two effects dominates depends on the extent of the network integration as demonstrated in the following example.\footnote{We can have a more integrated network by adding new links and (or) increasing the weight of existing links. For simplicity, we only add new links in this example.}

\begin{figure}[ht!]
	\centering
	\begin{subfigure}{0.3\textwidth}
		\includegraphics[width=\textwidth]{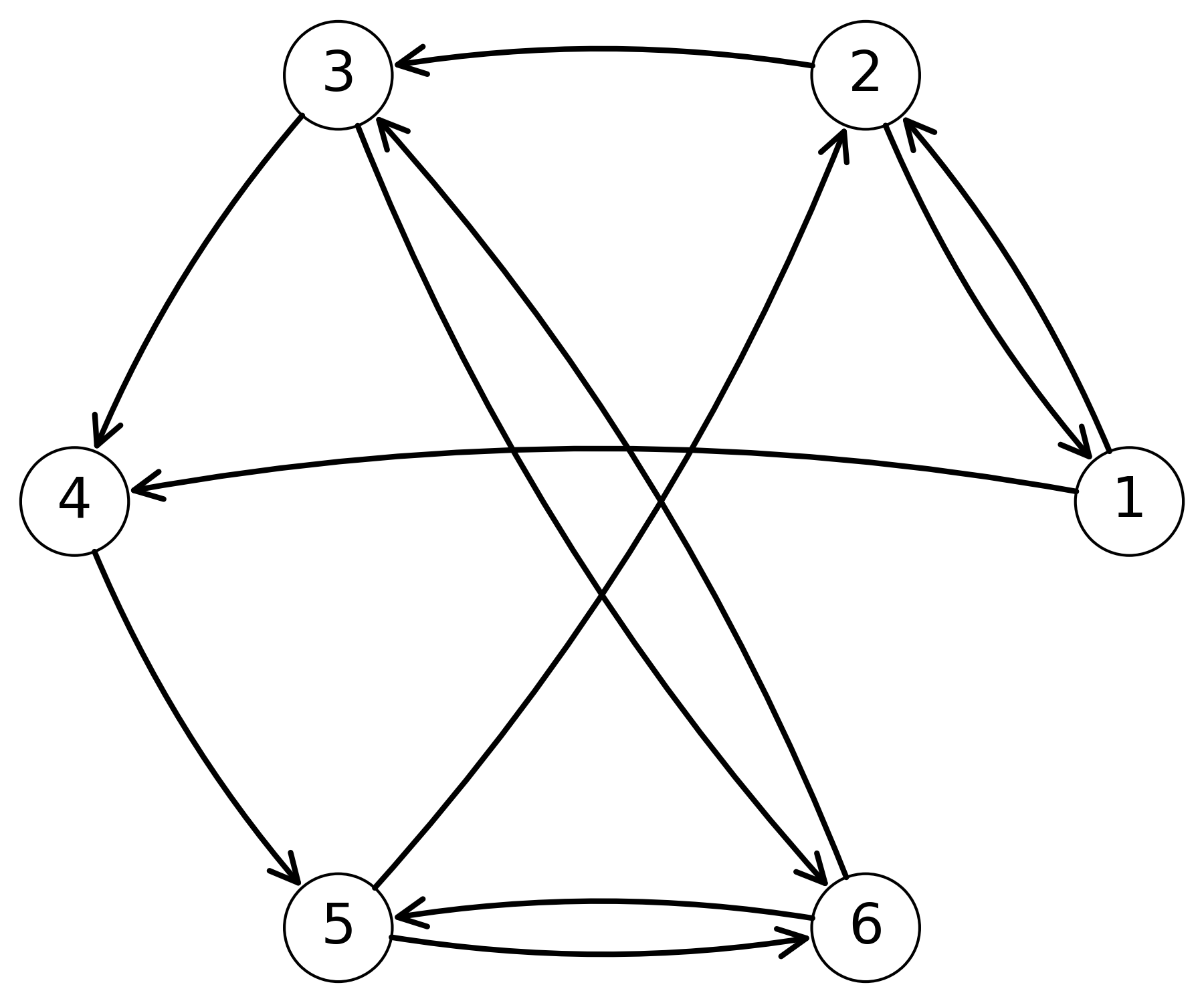}
		\caption{ $\bm{w}$}
		\label{fig:ex1}
	\end{subfigure}
	\quad
	\begin{subfigure}{0.3\textwidth}
		\includegraphics[width=\textwidth]{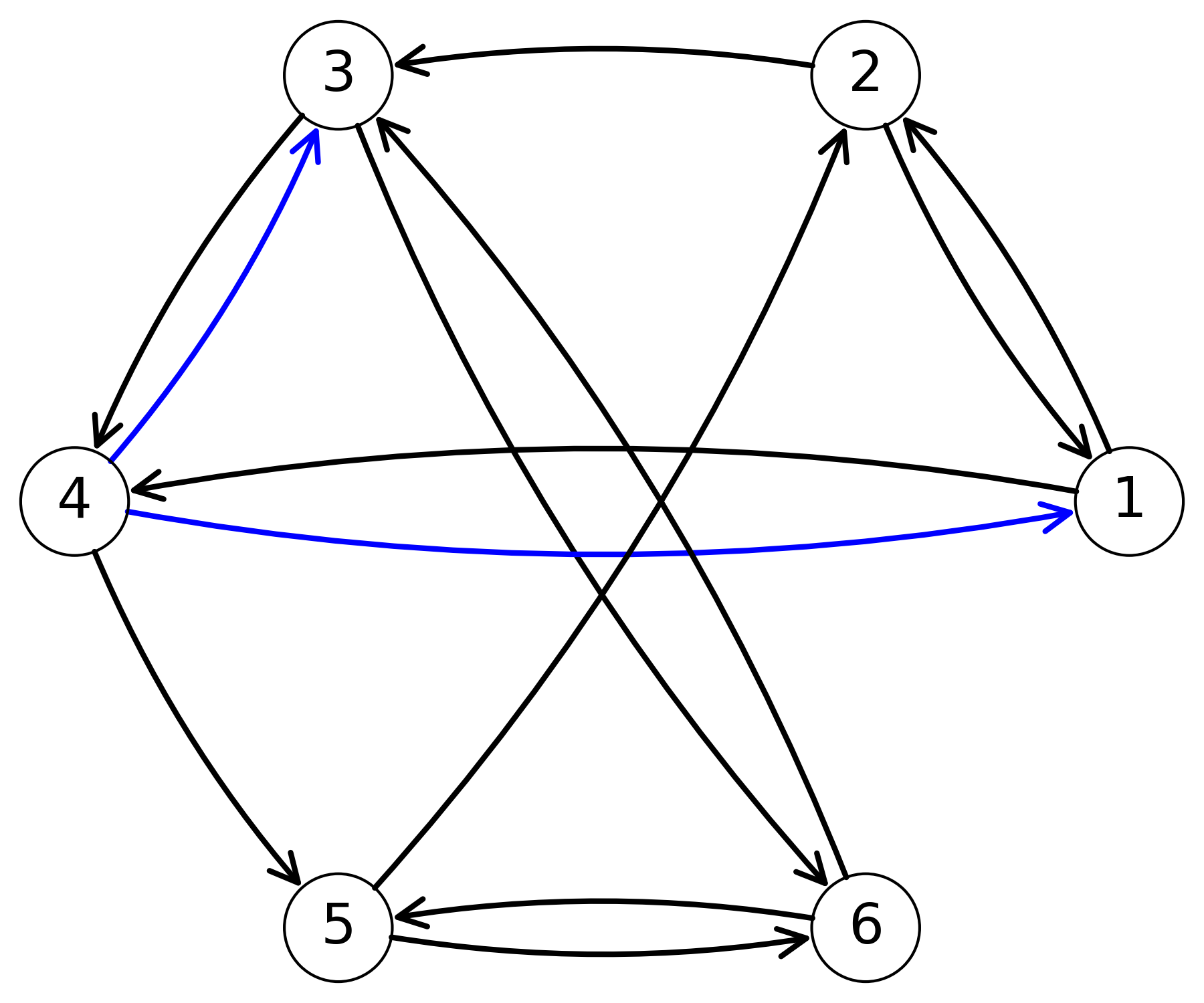}
		\caption{$\bm{w'}$}
		\label{fig:ex2}
	\end{subfigure}	
	\quad
	\begin{subfigure}{0.3\textwidth}
		\includegraphics[width=\textwidth]{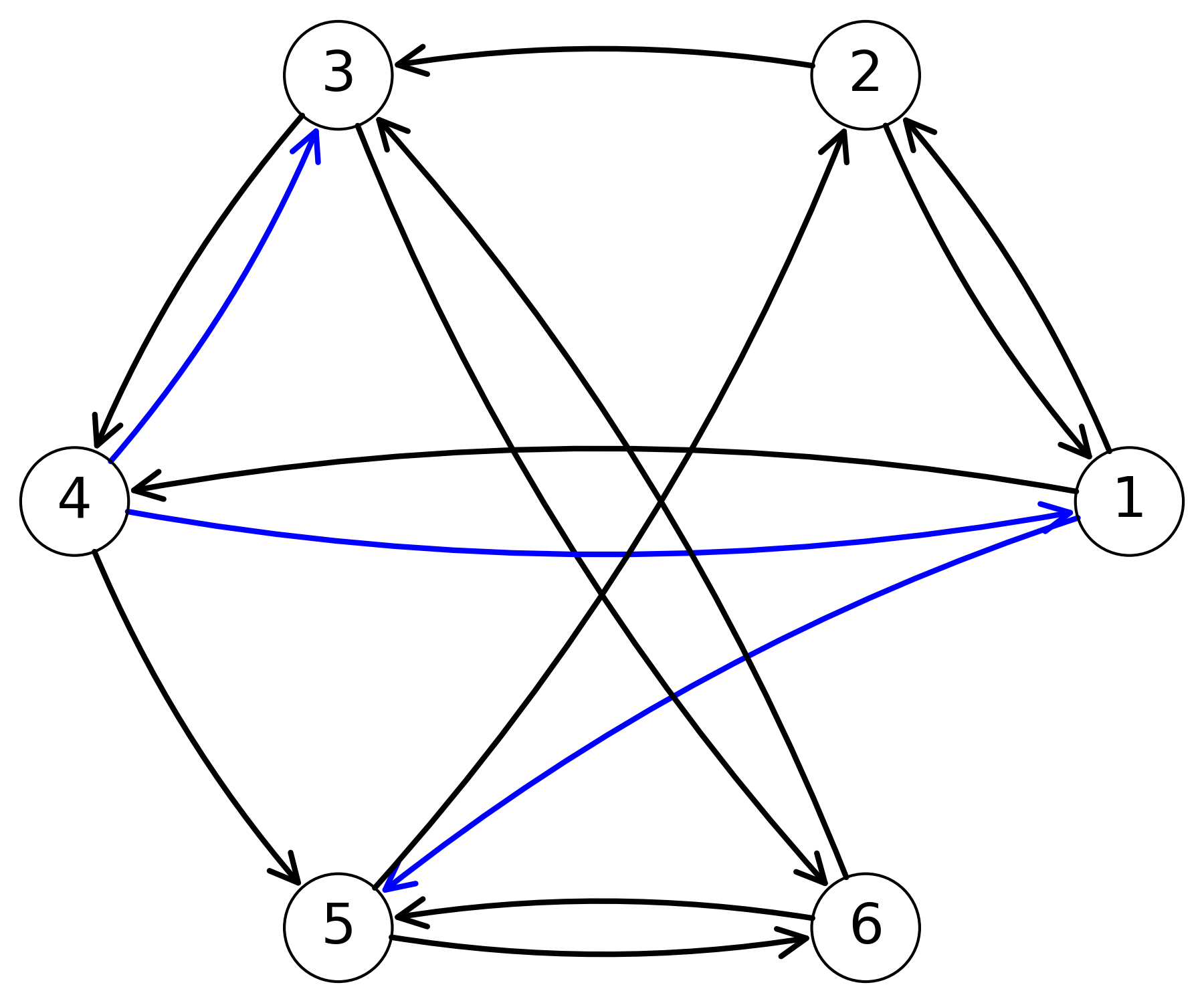}
		\caption{$\bm{w''}$}
		\label{fig:ex3}
	\end{subfigure}	
	\caption{An example of network integration}
	\label{fig:example}
\end{figure}

{\bf Example}: We introduce asymmetry favoring $b$ by setting $f(e_a)=e_a^{0.6}$ and $f(e_b)=1.1e_b^{0.6}$ and keeping other factors symmetric: $m_i = 1 ~\forall i, \beta = 2, k = 1.5$, and $c(e) = \exp(e)$. Consider first the case where we initially have network $\bm{w}$ with each link having weight 0.25 (Figure \ref{fig:ex1}).  Then, in the SPE outcome, $(e^*_a(\bm{w}), e^*_b(\bm{w}))= (0.27, 0.31)$, $X^*_a(\bm{w}) = 2.83$ and $X^*_b(\bm{w}) = 3.17$. Now, suppose that additional links of the same weight from user 4 are added to form $\bm{w'}$ (Figure \ref{fig:ex2}). Then, the stronger network effects increase the incentive for issuers to provide liquidity, especially for $b$ who issues the currency preferred by users. As a result, $b$ raises its commitment more than $a$ and gains a higher global adoption than when the network is less integrated. Specifically, in the new SPE outcome, $(e^*_a(\bm{w'}), e^*_b(\bm{w'}))= (0.29, 0.34)$, $X^*_a(\bm{w'}) = 2.79$ and $X^*_b(\bm{w'}) = 3.21$.

Now, we add an additional link into $\bm{w'}$ from user 1 to user 5, which gives us an even more integrated network $\bm{w''}$ (Figure \ref{fig:ex3}). If $a$ and $b$ kept the initial equilibrium commitments  corresponding to $\bm{w'}$ (i.e., $(0.29, 0.34)$), the higher network externalities would increase the usage of currency $b$ and decrease that of currency $a$.  However, the larger network externalities also mean that if $a$  tries to deter $b$, users would respond more to its commitment than before, giving it a higher utility than sharing the market with $b$. This change in the users' reaction gives $a$ more incentive to internationalize its currency for a monopoly status. As a result, $a$ deters $b$ by choosing $e^*_a(\bm{w''}) = 1$ and obtains the whole market.

\section{A $T$-currency problem}
\label{sec:dp}
Our focus on the two-currency framework is motivated by the rivalry between the US and China. While this dyadic competition provides a benchmark to analyze network effects and issuers' strategies, 
in this section, we extend the model to a setting with finitely many currencies. We  provide a solution concept for the currency allocation problem of users, given the commitment levels of issuers. With a well-defined solution for the currency choice problem, we can find the SPE of the game by backward induction. The main insights from the two-currency case continue to hold in the extended model.

\subsection{A dynamic programming approach}

The  objective of users is to minimize their overall cost by allocating their transactions to $T$ currencies issued by $T$ issuers. Let the set of issuers be $E = \{1,2,..., T\}$.
First, we assume that each user has an order of currencies that they consider for their transaction settlement and think of the currency allocation as choices in a recursive problem. We show below that under our assumptions, the optimal currency allocation for users is unique and order-independent. Let currencies in  consideration  for user $i$ be in order $1, 2,..., T$. The state variable is the user's remaining transactions to be settled in each stage $t$, denoted as $m_{it}$, and the choice variable is the user's usage of each currency $t$, denoted as $x_{it}$. Given state $m_{i,t-1}$, user $i$ chooses its usage of currency $t$, thus decides $m_{it} = m_{i,t-1} - x_{it}$. Since user $i$ needs to settle a total of $m_i$ transactions, we set $m_{i0} = m_i$.
The perceived cost of using $x_{it}$ of currency $t$ for user $i$ is
\begin{equation*}
	\cC_{it} (\bm{x_t})
	=\frac{\beta}{2} x_{it}^2
	- w_{ij} x_{it} x_{jt}
	- f_{it} x_{it}
\end{equation*}
where $x_{it} = m_{i,t-1} - m_{it} $.
We denote the vector of usage of currency $t$ as $\bm{x_t}$.
The cost minimization problem of each user $i$ is
\begin{equation*}
	\label{total_c}
\min_{ x_{i1},..., x_{iT}} 
	\vV_i (\bm{x_1}, ..., \bm{x_T} ) = \sum_{t=1}^{T}\cC_{i}(\bm{x_t}) \quad
	 \text{s.t. } \sum_{t=1}^{T}x_{it} = m_i. 
\end{equation*}

We can then rewrite the problem in terms of the Bellman equation for each user $i$ as follows.
\begin{equation*}
	V_{i,t-1} (m_{i,t-1}) = \min_{0 \leq m_{it} \leq m_{i, t-1}} \{ \cC_{it}(m_{i,t-1} - m_{it}) + V_{it} (m_{it})\}
\end{equation*}
%
%
%
The FOC and the Envelope theorem yield the Euler equation
\begin{equation}  \label{euler_c}
	\cC_{it}'(x_{it}) = \cC_{i,t+1}'(x_{i, t+1}),
\end{equation}
which is equivalent to $
	\beta x_{it} - \sum_{j \neq i} w_{ij} x_{jt} - f_{it} =
	\beta x_{i, t+1} - \sum_{j \neq i} w_{ij} x_{j, t+1} - f_{i,t+1} $.
If $m_{it} \leq \sup \{m \in (0,m_{i0}] : \cC_{it}'(m) \leq \cC_{i,t+1}'(0) \}$, then the marginal cost of settling all the remaining transaction with currency $t$ is lower than the marginal cost of starting settling with $t+1$. Thus, user $i$ would choose to settle all the remaining transactions with currency $t$ without considering other currencies. 
Hence, when $\beta$ is sufficiently high such that $m_{i, t} \geq \sup \{m \in (0,m_{i0}] : \cC_{i,t}'(m) \leq \cC_{i,t+1}'(0) \}$ for all $i$ and for all $t$, the Euler equation (\ref{euler_c}) holds for all $i$ and for all $t$, which implies the solution is interior and the order of consideration of the currencies does not matter for the resulting usage of each currency by users.\footnote{When $\beta$ is sufficiently low, then some users may choose to settle all their transactions before considering $T$, thus the order of consideration of the currencies matters for the users' currency choice.} Therefore, we can identify the minimum value of each $\beta$, denoted $\underline{\beta}$, that ensures a unique, nonnegative, and order-independent solution for the users' currency choice problem.\footnote{It can be checked that the condition is satisfied when $\beta  \geq \max_{i \in N, (t, \tau) \in E^2} \frac{\sum_{j \neq i} w_{ij} m_{j} +f_{it} -f_{i \tau} } {m_i} $.} Suppose that $\beta \geq \underline{\beta}$. Then, we can rewrite the Euler equations in terms of matrices:
\begin{equation} \label{euler_matrix}
	\beta \bm{x_{t+1}} - \bm{w x_{t+1}} - \bm{f_{t+1}} 
	- \beta \bm{x_t} + \bm{w x_t} +  \bm{f_t} = \bm{0}.
\end{equation}
Hence, the ``law of motion" for currency usage is
\begin{equation} \label{lom}
	\bm{x_{t+1}}
	= \bm{x_t} + (\beta \mathbb{I} -  \bm{w})^{-1} (\bm{f_{t+1}} -  \bm{f_t}).
\end{equation}
For the currency usage to be well-defined, we need $\beta > r(\bm{w})$, where $ r(\bm{w})$ is the spectral radius of the weight matrix $\bm{w}$. 
From (\ref{lom}), each usage vector $\bm{x_t}$ can be written as a function of $\bm{x_1}$. Using the constraint $\sum_{t=1}^{T} \bm{x_t} = \bm{m}$ and some linear algebra, we can derive the unique solution for the usage of each currency by all users given the commitment levels of currency issuers.

So far, we have shown how the optimal currency allocation can be derived given the commitment levels of issuers. We can then find the SPE of game among $T$ issuers by backward induction. The following corollary shows that the main insights from the two-currency case continue to hold in this extended model.

\begin{corollary} \label{cor:T} \leavevmode
	\begin{itemize} 
		\item [(i)] Suppose that players are symmetric. As $k$ increases, there are less global currencies.
		\item[(ii)] Suppose that there are $T$ users, each prefers a distinct currency, i.e., $f_{ip}(e) = \mu f_{i,-p}(e)$. Let all other factors be symmetric. The issuer of the currency that is preferred by the user with a higher centrality commits more strongly and captures a larger market share.
		\item [(iii)] Suppose that players are symmetric. When the network is more integrated, there are less global currencies.
	\end{itemize}
\end{corollary}
The mechanisms behind these results are analogous to those discussed in Section \ref{sec:results}. 

\subsection{Usage difference and the Gini coefficient}
In this section, we analyze the relationship between the currency choice of each user to their position in the network. Define $\bm{\gamma}_t \coloneqq \frac{1}{\beta} [\bm{f_t} - \bm{f_{t+1}} ]$, which captures the difference in marginal costs between $t$ and $t+1$, excluding the network effects. From equation  (\ref{lom}), the extent to which users choose currency $t$ over currency $t+1$ is captured by the usage difference
\begin{equation}
	\bm{d_t} \coloneqq  \bm{x_t} -  \bm{x_{t+1}} 
	= \left( \mathbb{I} -  \frac{1}{\beta} \bm{w} \right)^{-1}  \bm{\gamma}_t,
\end{equation}
which corresponds to an adjusted Katz-Bonacich centrality.
In fact, due to the order-independence of the solution, we can write the usage difference between any two currencies in an analogous way. That is, the difference in usage of any two currencies is an adjusted Katz-Bonacich centrality,
where the connection from a user to a trade partner is weighted by the difference in marginal costs that the partner faces when choosing between currencies. This gives us an insight into the measure of the diversity in each user's currency choice.
Let the Gini coefficient, originally introduced in \cite{gini1912variabilita}, measure the extent to which the allocation of usage among currencies deviates from a perfectly equal allocation. The coefficient is given by  the following formula
\begin{equation} \label{gini}
	G_i = \frac{\sum_{t=1}^{T} \sum_{u=1}^{T} |x_{it} - x_{iu}|}{2T\sum_{t=1}^{T} x_{it}}
	= \frac{\sum_{t=1}^{T} \sum_{u=1}^{T} |x_{it} - x_{iu}|}{2Tm_i}.
\end{equation}

\begin{proposition} \label{prop:katz_gini}
	Let $\tilde{\kappa}_{it} \coloneqq \frac{t(T-t) |d_{it}|}{Tm_i}$ be user $i$'s Katz-Bonacich centrality adjusted by its difference in marginal cost between $t$ and $t+1$ without network effects. Then, $G_i = \sum_{t=1}^{T-1} \tilde{\kappa}_{it}$.
\end{proposition}

Proposition \ref{prop:katz_gini} states that the Gini coefficient of each user's currency allocation is a sum of Katz-Bonacich centralities adjusted by the difference in marginal costs for using currencies, excluding network effects. In other words, a country's currency allocation for trade settlement is more concentrated when it has a higher network centrality. 
This result highlights that users with more or stronger trade ties in the global trade network tend to concentrate its usage more heavily on a single or a few currencies. 
The intuition is similar to Lemma \ref{lem:d_kappa} for the two-currency case: a user's central position makes the gains from using a common medium of exchange more valuable, thereby reducing the incentives to diversify across currencies. We demonstrate this concentrated currency allocation in the following example.

\begin{figure}[ht!]
	\centering
	\includegraphics[width=0.3\textwidth]{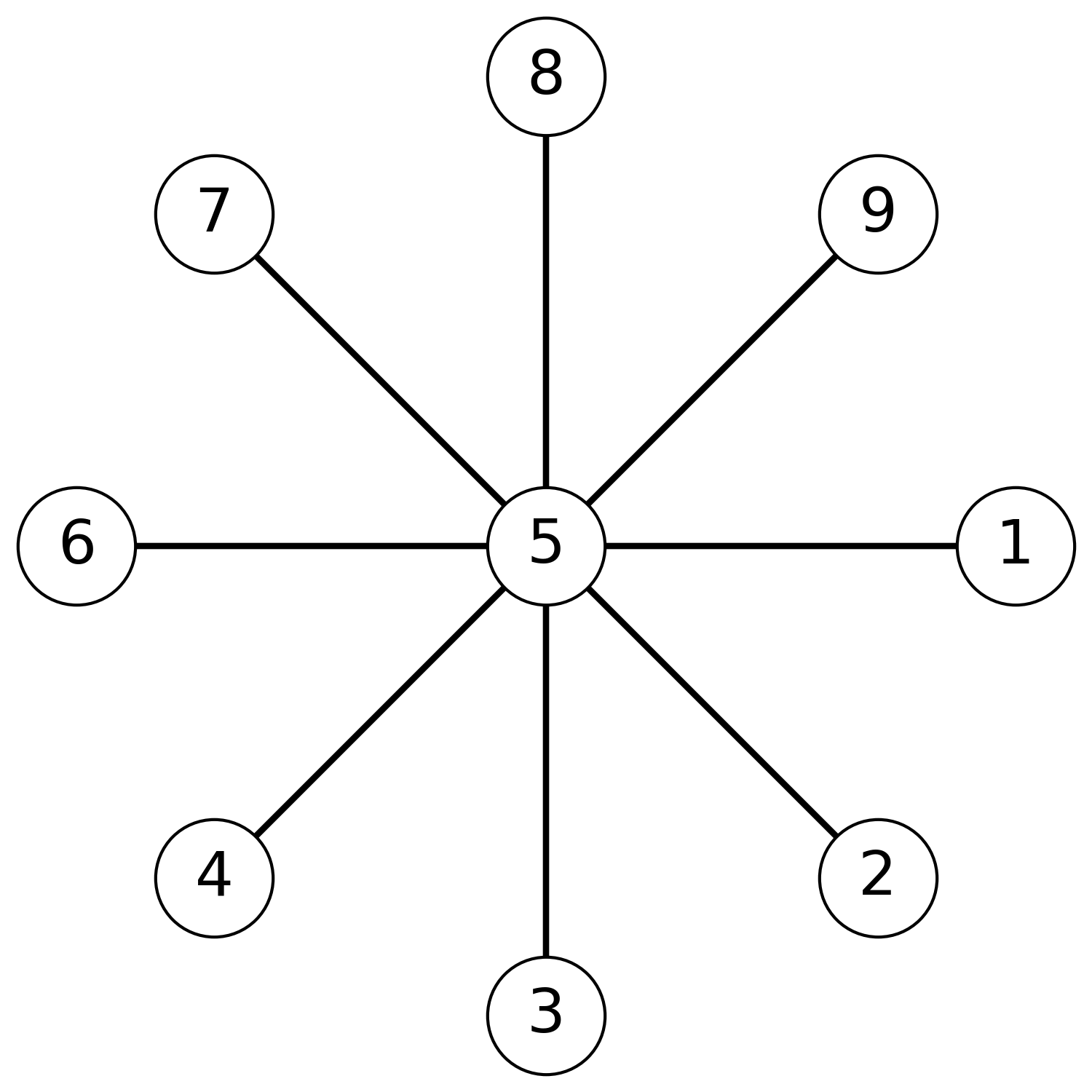}
	\caption{A star network}
	\label{fig:ex4}
\end{figure}

{\bf Example}: Consider the undirected network as depicted in Figure \ref{fig:ex4}, where each link has a weight of 0.125. User 5 is the most central node, while all other users have identical (lower) centrality. Suppose that there are five issuers (1 through 5) competing to attract global currency usage, and the commitments by issuers that enter later generate lower liquidity for users due to additional transaction and management costs for adopting new currencies into the payment system. Specifically, let the liquidity function for currency 1 be  $f_1(e) = e^{0.4}$ and $f_t(e) = 0.7^{t-1} f_{t-1}(e)$ for $t \in \{2,3,4,5\}$. Let $\beta = 2.1$, $k = 0.3$, and $c(e)=\exp (1.5 e)$. In the SPE outcome of this game, issuers that enter later choose lower commitment, and the last issuer chooses not to internationalize its currency. Specifically, the SPE outcome commitments are given by
$(e_1^*, e_2^*, e_3^*, e_4^*, e_5^*) = (0.86, 0.7, 0.54, 0.42, \emptyset)$. While facing identical preference, users' currency allocations differ depending on their network centrality. As shown in Figure \ref{fig:ex4_result}, the most central user (user 5) concentrates its usage on currency 1, while the others allocate their usage more evenly across multiple currencies. Indeed, the Gini coefficient of user 5's currency allocation is 0.42, whereas that of the other users is 0.3.

\begin{figure}[ht!]
	\centering
	\includegraphics[width=0.5\textwidth]{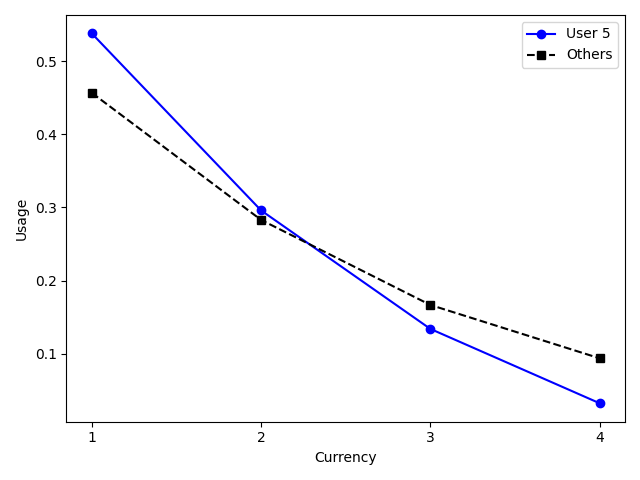}
	\caption{Currency usage: User 5 vs others}
	\label{fig:ex4_result}
\end{figure}

The above result has important implications. It suggests that central countries contribute disproportionately to reinforcing global currency dominance by concentrating their usage, affecting their trade partners' currency choices. In turn, this can deepen the entrenchment of incumbent currencies and raise the bar for challengers, particularly if central countries are biased toward the incumbent or resistant to adopting emerging currencies. Our model yields testable implications, for instance, it predicts a systematic relationship between trade network centrality and the degree of currency concentration across countries. It can thus serve as a guide for empirical studies examining how the structure of global trade affects international currency use over time.

\section{Conclusion}

We contribute to understanding the rise and fall of dominant international currencies by
developing a game-theoretic model where issuers and users interact in the presence of network spillover. The objective of issuers is to maximize the international usage of their currencies minus the cost of providing liquidity. Providing liquidity is costly but promotes the international usages of currencies. On the other hand,  the objective of users is to minimize the exchange rate risks associated with settling a given trade volume by choosing available international currencies. We assume that the exchange rate risks can be reduced by choosing a currency that trade partners choose and by the liquidity of that currency. 
Our framework is generalized to a setup with finite sets of issuers and users by the application of a dynamic programming approach.  

Our analysis reveals structural barriers to the internationalization of an emerging currency, explaining why being a dominant trade partner in bilateral trade is insufficient to achieve dominant currency status.
We find that the incumbent issuer has an advantage as it can deter the emerging currency by providing enough liquidity. This advantage is amplified when the trade network is more integrated. 
Here, the central users play a critical role as it is their preferences that matter most for choices of both users and issuers. We find that more central users concentrate their usage of currencies, while more peripheral  users diversify theirs. The issuer of a currency preferred by more central users increases the liquidity, which in turn leads to greater adoption of the currency. 
This feedback mechanism leads to our central insight on the interplay of liquidity provision and currency choice of users that determines the dominant currency. 

Our model can be used to explain 1) the US's first-mover strategy---establishing the Bretton Woods Institutions---to maintain the USD dominance and why an emerging issuer, such as China, faces prohibitive costs to establish itself as an issuer of a new global currency,
2) the increasing reliance on the USD for invoicing and settlement as international trade becomes more integrated, 3) how the US's high centrality and domestic user bias toward the USD incentivize global provision of liquidity. The mechanisms also elucidate China's restricted liquidity provision and the RMB's limited global adoption despite its trade dominance.

\clearpage
\appendix
\section*{Appendix: Proofs}

\begin{proof}[Proof of Lemma \ref{lem:d_kappa}]
	The proof follows directly from \eqref{foc:d} and \eqref{katz}.
\end{proof}

\begin{proof}[Proof of Lemma \ref{lemma:SPE}]
	To show the existence of the SPE, we show that the best responses exist at every subgame.
	First, consider $b$'s best response given any $e_a$ chosen by $a$.
	Define $\bm{A} \coloneqq  \sum_{\ell=0}^{\infty} \frac{1}{\beta^{\ell}} \bm{w}^{\ell}$. From (\ref{foc:usage}), the total usage of currency $b$ is $X_b (e_b) = \frac{M - \sum_{i \in N} \sum_{j \in N} A_{ij} \gamma_j (e_a, e_b)}{2}$. Thus, $b$'s utility is
	$u_b (e_b) =  \frac{M - \sum_{i \in N} \sum_{j \in N} A_{ij} \gamma_j (e_a, e_b)}{2} - k c(e_b),$	%
	which is concave as $ -\sum_{i \in N} \sum_{j \in N} A_{ij} \gamma_j(e_b)$  is a sum of concave functions and $c(e_b)$ is convex.  We find $b$'s best response to $a$'s commitment level $e_a \in [0,1]$ by taking the FOC w.r.t $e_b$
	\begin{equation} \label{BR_b}
		\frac{\partial u_b (e_a,e_b)}{\partial e_b} = \frac{1}{2 \beta} \sum_{j \in N} \sum_{i \in N} A_{ij} f_{jb}'(e_b) - k c'(e_b) = 0.
	\end{equation}
	Let $e_b'$ be the unique solution for (\ref{BR_b}). Then, $b$'s best response when $a$ internationalizes its currency is $\tilde{e}_b = \min \{1,  e_b'\}$ if $u_b(e_a, \tilde{e}_b) > 0$ and $\emptyset$ otherwise. When $a$ chooses not to internationalize its currency, $b$'s best response is $e_b = 0$.
	
	Next, consider when $a$ chooses an $e_a$ between [0,1] anticipating $b$'s best response, which is either $\tilde{e}_b$ or $\emptyset$. Since $u_b(e_a, \tilde{e}_b)$ is decreasing in $e_a$, there exists a value $\hat{e}_a \in \mathbb{R}$ such that $u_b(\hat{e}_a, \tilde{e}_b)=0$, thus $b$'s best response is $\tilde{e}_b$ if $e_a <\hat{e}_a$ and $\emptyset$ if $e_a \geq \hat{e}_a$. Hence, $a$'s utility function is 
	$u_a (e_a) =  \frac{M + \sum_{i \in N} \sum_{j \in N} A_{ij} \gamma_j (e_a, \tilde{e}_b)}{2} - k c(e_a)$ for  $e_a < \hat{e}_a$, and
	$u_a (e_a) =  M - k c(e_a)$ for  $e_a \geq \hat{e}_a$. Since $\lim_{e_a \rightarrow \hat{e}_a^-} = \frac{M + \sum_{i \in N} \sum_{j \in N} A_{ij} \gamma_j (\hat{e}_a, \tilde{e}_b)}{2} - k c(\hat{e}_a)  \leq M - k c(\hat{e}_a)$, $u_a (e_a)$ is upper-semicontinuous. By the Weierstrass maximum theorem, $a$'s choice set obtains its maximum values. $a$ compares the optimal commitment level $e_a$ and the option of not to promote its currency ($u_a(\emptyset) = 0$) and chooses its best response.
	Therefore, the best responses exist at every subgame and thus, the game has at least one SPE.
\end{proof}

\begin{proof}[Proof of Lemma \ref{lemma:k}]
	To analyze the SPE, we derive the best response of each issuer given the choice of the other. 
	Given $b$'s best response as derived in the Proof of Lemma \ref{lemma:SPE}, we can  analogously find $a$'s optimal commitment level $\tilde{e}_a$ when they share the market. However, it may be more beneficial for $a$ to choose a commitment level that makes $b$ give up. The commitment level for $a$ to deter $b$ is $\underline{e}_a = \hat{e}_a$ if $\hat{e}_a \geq 0$ and  $\underline{e}_a =0$ otherwise, which is not feasible if $\hat{e}_a > 1$. 
	$a$'s best response is $\underline{e}_a$ if $ u_a(\tilde{e}_a, \tilde{e}_b) \leq u_a(\underline{e}_a, \emptyset)$ (deterring is more beneficial) and $\hat{e}_a \leq 1$ (deterring is feasible). When it is not beneficial or not feasible for $a$ to deter $b$, its best response is $\tilde{e}_a$ if $u_a(\tilde{e}_a, \tilde{e}_b) > 0$, and $\emptyset$ otherwise.
	Since both issuers play their best responses in the SPE, possible SPE outcomes are
	\begin{equation*}
		(e_a^*, e_b^*) =  \left\{ \begin{array}{lcl}
			 (\underline{e}_a, \emptyset) & \mbox{if} &  u_a(\tilde{e}_a, \tilde{e}_b) \leq u_a(\underline{e}_a, \emptyset) \text{ and } \hat{e}_a \leq 1
			\\
			 (\emptyset, 0) & \mbox{if} &  u_a(\tilde{e}_a, \tilde{e}_b) \leq 0 \text{ and either } u_a(\tilde{e}_a, \tilde{e}_b) > u_a(\underline{e}_a, \emptyset) \text{ or } \hat{e}_a > 1
			\\
			(\tilde{e}_a, \tilde{e}_b) &  &  \mbox{otherwise.} 
		\end{array}\right.
	\end{equation*}
	Now we are ready to find the values of $k$ for each case above. 
	
	We first find the values of $k$ so that $(e^*_a, e^*_b) = ( \underline{e}_a,  \emptyset)$.
	We can check that when $k\to0$,  $(e^*_a, e^*_b) =(\tilde{e}_a(k), \tilde{e}_b(k)) =  (1,1)$ and both gain a non-negative total usage, thus $\hat{e}_a (k ) > 1$.
	Moreover, when $k=\frac{M}{c(0)}$, $u_b(0, \tilde{e}_b) < 0$, thus $\hat{e}_a < 0$ and $(e_a^*, e_b^*) = ( 0, \emptyset )$. 
	By the chain rule, $u_b(\tilde{e}_b;k)$ is decreasing in $k$, thus $\hat{e}_a$ is decreasing in $k$. Therefore, $\exists\hat{k}\in[0,\frac{M}{c(0)}]$ such that $\hat{e}_a(\hat{k}) = 1$ and $\hat{e}_a(k) < 1$ for any $k \in (\hat{k}, \frac{M}{c(0)}]$.
	Similarly, $\exists\check{k}\in[0,\frac{M}{c(0)}]$ such that $\hat{e}_a(\check{k}) = \tilde{e}_a(\check{k})$  and $\hat{e}_a(k) < \tilde{e}_a(k)$ for any $k \in (\check{k}, \frac{M}{c(0)}]$.
	Next, we find the values of $k$ such that $u_a(\tilde{e}_a, \tilde{e}_b; k) \leq u_a(\underline{e}_a, \emptyset; k)$. We only need  to consider the range of $k$ such that $\tilde{e}_a (k) = e'_a(k) \leq 1$ and $\hat{e}_a(k) \geq 0$.
	By applying the chain rule, 
	$\frac{\partial [u_a(\underline{e}_a, \emptyset; k) -  u_a(\tilde{e}_a, \tilde{e}_b; k)]}{\partial k} = 
	-k \left[\frac{\partial c(\underline{e}_a)}{\partial \underline{e}_a} \frac{\partial \underline{e}_a}{\partial k} - \frac{\partial c(\tilde{e}_b)}{\partial \tilde{e}_b} \frac{\partial \tilde{e}_b}{\partial k} \right]  
	- [c(\underline{e}_a) - c(\tilde{e}_a)]$,
	which is negative when $\underline{e}_a(k) > \tilde{e}_a(k)$ and is positive when $\underline{e}_a(k) < \tilde{e}_a(k)$, that is, when $k > \check{k}$.
	Note that when $k \rightarrow 0$, $u_a(\underline{e}_a, \emptyset) \rightarrow M > u_a(\tilde{e}_a, \tilde{e}_b)$; when $\hat{e}_a(k) = 0$, $u_a(\underline{e}_a, \emptyset) > u_a(\tilde{e}_a, \tilde{e}_b)$. Hence, $u_a(\tilde{e}_a, \tilde{e}_b)$ and $u_a(\underline{e}_a, \emptyset)$ may intersect at at most two points, and $u_a(\tilde{e}_a, \tilde{e}_b) \leq u_a(\underline{e}_a, \emptyset)$ for any $k$ larger than the highest intersection point. (If they do not intersect, this inequality is always satisfied.)
	Note that in the case the two functions intersect at two points, $\hat{k}$ has to be higher than the lower intersection point. This is because $\hat{e}_a(\hat{k}) = 1 > \tilde{e}_b(\hat{k})$, which implies $\left. \frac{\partial [u_a(\underline{e}_a, \emptyset; k) -  u_a(\tilde{e}_a, \tilde{e}_b; k)]}{\partial k} \right|_{k=\hat{k}} > 0 $.
	%
	Let $\dot{k} \geq \hat{k}$ be the minimum $k$ such that both inequalities hold for any $k \geq \dot{k}$. If $k \in [\dot{k}, \frac{M}{c(0)}]$, then $(e^*_a, e^*_b) = ( \underline{e}_a,  \emptyset )$.
	%

	Next, we find the values of $k$ such that $(e^*_a, e^*_b) = (\emptyset, 0 )$, which occurs when $a$  gets a negative utility by playing $\tilde{e}_a$ and not being able to deter $b$. By the chain rule, $u_a (\tilde{e}_a, \tilde{e}_b; k)$ is decreasing in $k$. Consider $k'$ such that $u_a (\tilde{e}_a, \tilde{e}_b; k') = 0$. Such $k'$ is in $[0,  \frac{M}{c(0)})$ since $u_a(\tilde{e}_a, \tilde{e}_b;k) < 0$ when $k =  \frac{M}{ c(0)}$.
	Then, $u_a(\tilde{e}_a, \tilde{e}_b;k) \leq 0$ given any $k \in [k', \frac{M}{c(0)}]$. When $k' \geq \dot{k}$, $\nexists k$ such that $(e^*_a, e^*_b) = (\emptyset,  0)$. When $k' < \dot{k}$,  if $k \in [k', \dot{k})$, then $(e_a^*, e_b^*) = (\emptyset, 0)$. 
	Let $\underline{k} = \min \{k', \dot{k}\}$. If $ k \in [\underline{k}, \frac{M}{c(0)}]$, then either $e_a^*=\emptyset$ or $e_b^*=\emptyset$. If $k \in [0, \underline{k})$, then  $(e_a^*, e_b^*) = ( \tilde{e}_a, \tilde{e}_b )$.
\end{proof}

\begin{proof}[Proof of Proposition \ref{prop:1stmover}]
	First, we show that when players are symmetric, $\nexists k$ such that  $(e^*_a, e^*_b) = (\emptyset,  0)$.
	When players are symmetric, $\tilde{e}_a(k) = \tilde{e}_b(k) = \tilde{e}(k)$ for any $k$.
	Consider $k'$ such that $u_a (\tilde{e}, \tilde{e}; k') = u_b (\tilde{e}, \tilde{e}; k') = 0$ as defined in the proof of Lemma \ref{lemma:k}. 
	From (\ref{foc:d}), $d_i(\tilde{e}(k'), \tilde{e}(k')) = 0$ for all $i$, thus
	 $X_a (\tilde{e}(k'), \tilde{e}(k')) = \frac{M}{2}$. Hence, $a$ chooses $\underline{e}_a = \tilde{e}(k')$ and deters $b$. Therefore, $k' \geq \dot{k}$. This implies that $\nexists k$ such that  $(e^*_a, e^*_b)  = (\emptyset,  0)$.
	 
	Next, we derive the case when $a$ deters $b$ with zero commitment level. Let $\bar{k}$ be $k$ such that $\hat{e}_a(\bar{k}) = 0 $.  Since $\hat{e}_a$ is decreasing in $k$,  if $\bar{k} \leq k  \leq \frac{M}{c(0)}$, then $\{ e_a^*, e_b^*\} = \{0, \emptyset\}$ and $\{X_a^*, X_b^*\} = \{M, 0\}$.
	By Lemma \ref{lemma:k},  
	if $k \in [\underline{k}, \bar{k})$, then 
	$0 < e_a^* \leq 1 , e_b^* = \emptyset, X_a^* = M$ and $X_b^* = 0$.
	If $k \in [0, \underline{k})$, then 
	$e_a^*= e_b^* = \tilde{e}(k) \in [0, 1]$. 

	Next, we show that $\underline{k}$ is increasing in $\beta$ by proving both $u_a(\tilde{e}, \tilde{e}) - u_a(\underline{e}_a, \emptyset) = \frac{-M}{2} + k[c(\underline{e}_a(k, \beta))-c(\tilde{e}(k, \beta))]$ and $\beta$, $u_b(1, \tilde{e})$ are increasing in $\beta$. Since $\underline{k} \leq k'$,
	 $u_a(\tilde{e},\tilde{e}; \underline{k}) = u_b(\tilde{e},\tilde{e}; \underline{k}) \geq 0$, which implies $\underline{e}_a (\underline{k}, \beta) = \hat{e} (\underline{k}, \beta) \geq \tilde{e} (\underline{k}, \beta)$.
	From (\ref{BR_b}), $\tilde{e}$ is increasing in $\beta$. Consider any $k \in [\underline{k}, k')$. As $\beta$ increases, $\hat{e}_a (k, \beta)$ (and thus $\underline{e}_a(k, \beta) $) increases more than $\tilde{e}(k, \beta)$ so that $u_b(\hat{e}_a, \tilde{e}) = 0$, implying $k[c(\underline{e}_a(k, \beta))-c(\tilde{e}(k, \beta))]$ is increasing in $\beta$. 
	Since $d_i(1, e_b) \geq 0$ is decreasing in $\beta ~\forall i$, by the envelope theorem, $u_b(1, \tilde{e})$ is increasing in $\beta$. Therefore, $\underline{k}$ is increasing in $\beta$.
	
	Finally, we show that $\underline{k}$ is increasing in $M$. From (\ref{foc:d}) and (\ref{BR_b}), an increase in $M$ does not change $\bm{d}(e_a, e_b)$ and $\tilde{e}$, which means $u_b(e_a, \tilde{e})$ is increasing in $M$, implying $\underline{e}_a$ and $u_b(1, \tilde{e})$ are increasing in $M$. Hence,  both $k(c(\underline{e}_a(k))-c(\tilde{e}(k))]$ and  $u_b(1, \tilde{e})$ are increasing in $M$, implying $\underline{k}$ is increasing in $M$.
\end{proof}

\begin{proof}[Proof of Lemma \ref{lemma:dX}]
	Suppose that $\bm{w'}  \gneq \bm{w}$. Define $\bm{A} \coloneqq \sum_{\ell=0}^{\infty} \frac{1}{\beta^{\ell}} \bm{w}^{\ell}$ and  $\bm{A'} \coloneqq \sum_{\ell=0}^{\infty} \frac{1}{\beta^{\ell}} \bm{w'}^{\ell}$.
	Since $A'_{ij} \geq A_{ij}$ for all $(i,j)\in N^2$ with at least one inequality, from (\ref{foc:d}) and (\ref{foc:usage}), for any $p \in \{a, b\}$,
	\begin{equation} \label{derriv:w}
		\frac{\partial X_p (\bm{w'})}{\partial e_p} - \frac{\partial X_p (\bm{w})}{\partial e_p}  =\frac{1}{2 \beta} \sum_{i \in N} \sum_{j \in N} (A'_{ij} - A_{ij}) \frac{\partial f_{jp}(e_p)}{\partial e_p}>0.
	\end{equation}

	%
	
	Suppose that $f_{ip}(e) = \mu_i f_{i,-p}(e) ~\forall i\in N$. From (\ref{derriv:w}), for any $e\in [0,1]$,
	\begin{align*}
		\frac{\partial X_p (\bm{w'})}{\partial e_p}  - \frac{\partial X_p (\bm{w})}{\partial e_p}  = \frac{1}{2 \beta} \sum_{i \in N} \sum_{j \in N} (A'_{ij} - A_{ij}) \mu_j  \frac{\partial f_{j,-p}(e_{-p})}{\partial e_{-p}}
		> \frac{\partial X_{-p} (\bm{w'})}{\partial e_{-p}}  - \frac{\partial X_{-p} (\bm{w})}{\partial e_{-p}}.
	\end{align*}
\end{proof}

\begin{proof}[Proof of Proposition \ref{prop:katz}]
	Suppose that $k < \underline{k}$. Then, we have $e_p^* = \tilde{e}_p(\bm{w})$ for any $p \in \{a, b\}$. Without lost of generality, suppose that  $f_{ia}(e) = f_{jb}(e) = \mu f_{ib}(e) = \mu f_{ja}(e)$ with $\mu > 1$. Consider vector $\bm{\gamma}$ in (\ref{foc:d}). We have 
	\begin{align*}
		\gamma_h(e_a, e_b) & = \frac{f(e_a) - f(e_b)}{2\beta} \coloneqq \bar{\gamma} \\
		\gamma_i(e_a, e_b) & = \bar{\gamma} + \Delta_i (e_a, e_b) \\
		\gamma_j(e_a, e_b) & = \bar{\gamma} + \Delta_j (e_a, e_b)
	\end{align*}
	where $\Delta_i (e_a, e_b) \coloneqq \frac{f_{ia}(e_a) - f(e_a) - f_{ib}(e_b) + f(e_b)}{2\beta}$ and $\Delta_j (e_a, e_b) \coloneqq  \frac{f_{ja}(e_a) - f(e_a) - f_{jb}(e_b) + f(e_b)}{2\beta}$ measure the extent to which $i$ prefers $a$ and $j$ prefers $b$ compared to other users.
	Then, (\ref{foc:d}) becomes $\bm{d} = [\kappa_g \bar{\gamma} + A_{gi} \Delta_i + A_{g_j} \Delta_j]_{g=1}^n$. Since $\sum_{g=1}^n A_{gi} = \sum_{g=1}^n A_{ig} = \kappa_i$, the sum of all entries of $\bm{d}$ gives
	\begin{equation} \label{asy:kappa}
		X_a(e_a, e_b) - X_b(e_a, e_b) = \sum_{g=1}^n \kappa_g \bar{\gamma} + \kappa_i \Delta_i (e_a, e_b) + \kappa_j \Delta_j (e_a, e_b).
	\end{equation}
	Suppose that $\kappa_i > \kappa_j$. From (\ref{asy:kappa}), $X_a > X_b$ when $e_a=e_b$. Moreover, from (\ref{foc:usage}) and (\ref{asy:kappa}), $\frac{\partial X_a}{\partial e_a} \big|_{e_a = e} - \frac{\partial X_b}{\partial e_b}\big|_{e_b = e} = \frac{1}{4 \beta} [(\kappa_i - \kappa_j) [(\mu-1)f'_{ib}(e)]] > 0$
	for any $e$. Hence, $\tilde{e}_a \geq \tilde{e}_b$. Therefore, $e^*_a \geq e^*_b$ (equality holds only if $e^*_a = e^*_b = 1$ when $k$ is sufficiently close to 0) and $X^*_a > X^*_b$.
\end{proof}

\begin{proof}[Proof of Proposition \ref{prop:w}]
	Consider the game with $k \in [0, k')$.
	%
	%
	When the players are symmetric, $\gamma_{i} (e_a, \tilde{e})=\gamma (e_a, \tilde{e}) = \frac{f(e_a)-f(\tilde{e}(k))}{\beta} > 0$ for any $e_a \in (\tilde{e}(k),1]$ and for any $i$. Hence, $d_i = \gamma(e_a, \tilde{e}) \kappa_i > 0$ for any $i$. From (\ref{max_u}) and (\ref{foc:usage}), $u_b(e_a, \tilde{e}) = \frac{M-  \gamma(e_a, \tilde{e}) \sum_{i \in N} \kappa_i}{2} - kc(\tilde{e}(k)$. Since $\tilde{e}(k)$ is $b$'s best response to $e_a$, by the envelope theorem
	\begin{equation*}
		\frac{\partial u_b( e_a, \tilde{e})} {\partial w_{gh}} = - \frac {\gamma(e_a, \tilde{e})} {2} \sum_{i \in N} \frac{\partial  \kappa_i}{\partial w_{gh}} < 0
	\end{equation*}
	for any pair  $(g, h)$. 
	Since $\frac{\partial u_b( e_a, \tilde{e})} {\partial w_{gh}}<0$ holds for  any $e_a \in (\tilde{e}(k),1]$,  $\underline{e}_a(\bm{w'}, k) < \underline{e}_a( \bm{w}, k)$, which means $a$ can deter $b$ with a lower commitment level. 
	From Lemma \ref{lemma:dX}, $\tilde{e}(\bm{w'})>\tilde{e}(\bm{w})$ and thus  $u_a(\tilde{e}, \tilde{e}) - u_a(\underline{e}_a, \emptyset)$ is decreasing in $w_{gh}$. Therefore, $\underline{k}$ is decreasing in $w_{gh}$.
\end{proof}

\begin{proof}[Proof of Corollary \ref{cor:T}]
	(i) 
	Following analogous arguments as in the proof of Proposition \ref{prop:1stmover}, when $k$ is close to 0, all issuers plays $\tilde{e}(k)$ and share the market equally. Given any choices of the first $T-2$ issuers, issuer $T-1$ has a deterring commitment level $\underline{e}_{T-1}$ that is decreasing in $k$ and $e_t$ for any $t<T-1$. Similarly, issuer $T-2$ has deterring commitment levels $\underline{e}_{T-2}^1$ (with which $T-2$ and $T-1$ - choosing $\underline{e}_{T-1}$ - deter $T$) and $\underline{e}_{T-2}^2$ (with which $T-2$ deters both $T-1$ and $T$) that are decreasing in $k$, and so forth. Hence, when $k$ is higher, less issuers deter the entry of new currencies. Note that each of them may choose different commitment levels, for example, the first mover chooses a commitment level just enough for the second one to deter the remaining currencies. It follows that we may have asymmetric market shares when incumbent issuers decide their commitment levels to deter the emerging issuers.
	
	(ii) When there are $T$ currencies, (\ref{asy:kappa}) holds for any pair of currencies. This means users are more responsive to the commitment by issuers of the currencies prefered by more central users, which reinforces such issuers to commit strongly.
	
	(iii) When there are $T$ currencies, by analogous arguments as in the proof of Proposition \ref{prop:w}, $\frac{\partial u_s( e_t, \tilde{e})} {\partial w_{gh}} < 0$ for  any $e_t \in (\tilde{e}(k),1]$ and for any $s > t$, meaning that given the choices of first $t-1$ issuers,  $t$ can deter the remaining issuers with a lower commitment level. Thus, for any given $k$, the number of global currency issuers (weakly) decreases as the network becomes more integrated.
\end{proof}

\begin{proof} [Proof of Proposition \ref{prop:katz_gini}]
	Since $x_{i\tau} - x_{i, \tau+c} = \sum_{t=\tau} ^{\tau+c-1} d_{it} $  for any natural number $c \in [1, T-\tau]$, the statement follows directly from \eqref{gini}.
\end{proof}

\bibliographystyle{ecta}
\bibliography{nw_currency} 
\end{document}